\documentclass{article}

\PassOptionsToPackage{numbers, compress}{natbib}

\usepackage[main, final]{neurips_2026}

\usepackage[utf8]{inputenc}
\usepackage[T1]{fontenc}
\usepackage{hyperref}
\usepackage{url}
\usepackage{booktabs}
\usepackage{amsfonts}
\usepackage{nicefrac}
\usepackage{microtype}
\usepackage{xcolor}

\usepackage{graphicx}
\usepackage{amssymb}
\usepackage{mathtools}
\usepackage{amsthm}
\usepackage{tikz}
\usepackage{amsmath}
\usepackage{colortbl}
\usepackage{eucal}
\usepackage{bm}
\usepackage{subfigure}
\usepackage{multirow}
\usepackage{enumitem}
\usepackage{xspace}
\usepackage[english]{babel}
\usepackage{blindtext}
\usepackage{pifont} 
\usepackage{stackengine}
\usepackage{etoc}
\usepackage{tabularx}
\usepackage{setspace}
\usepackage{wrapfig}
\usepackage{float}
\usepackage{algorithm}
\usepackage{algpseudocode}

\newtheorem{corollary}{Corollary}
\newtheorem{proposition}{Proposition}
\newtheorem{definition}{Definition}
\newtheorem{assumption}{Assumption}
\newtheorem{remark}{Remark}

\DeclareMathOperator*{\argmin}{arg\,min}

\newcolumntype{C}[1]{>{\centering\arraybackslash}p{#1}}
\newcolumntype{L}[1]{>{\raggedright\arraybackslash}p{#1}}

\newcommand{\withours}[1]{\textbf{Ours}\textsubscript{#1}}

\def\ourSystem{\text{RxGS}\xspace}
\def\ie{\textit{i.e.}\@,\xspace}
\def\eg{\textit{e.g.}\@,\xspace}

\newlength{\mylen}
\hypersetup{
    breaklinks=true,
    colorlinks=true,
    citecolor=black,
    linkcolor=black,
    urlcolor=black,
    hypertexnames=false
}

\hypersetup{citecolor=black,linkcolor=black,urlcolor=black}

\title{\ourSystem: Receiver-Generalizable 3D Gaussian Splatting for Radio-Frequency Data Synthesis}

\author{
  Kang Yang \hspace{2em} Mani Srivastava \\
  Department of Electrical and Computer Engineering\\
  University of California, Los Angeles\\
  \texttt{kyang73@g.ucla.edu \quad mbs@ucla.edu}
}

\begin{document}

\maketitle

\addtocontents{toc}{\protect\setcounter{tocdepth}{-1}}

\begin{abstract}

Radio-frequency~(RF) data synthesis predicts the received signal given transmitter and receiver positions, and is essential for wireless applications.
Recent~3D Gaussian Splatting~(3DGS)-based methods achieve efficient synthesis at any transmitter location within a surveyed scene, but require a separate model for each receiver.
We therefore present~\ourSystem, a shared~3DGS model that first learns receiver-independent Gaussian geometry and then freezes it while learning directional radiance conditioned on receiver position.
Its global conditioning branch captures receiver-dependent effects shared across the scene, while a local branch models per-scatterer variations from the receiver's geometry and occlusion.
The shared geometry also supports batched rendering through a multi-receiver rasterizer.
Across~RF datasets,~\ourSystem predicts at receiver positions excluded from training within the surveyed scene while significantly reducing training, inference, and storage costs.
Project~page: \href{https://kangyangg.com/projects/rxgs/}{https://kangyangg.com/projects/rxgs/}.
\end{abstract}

\vspace{\mylen}
\section{Introduction}\label{sec_introduction}
\vspace{\mylen}

Radio-frequency~(RF) data, such as received signal strength indicator~(RSSI) and channel state information~(CSI), quantifies how~RF signals propagate from a transmitter to a receiver in a scene~\cite{ma2019wifi}.
Such data is essential for wireless networking and sensing applications~\NoHyper\cite{abedi2020witag, rappaport2002wireless, okubo2024integrated, zhao2023nerf2, lu2025geraf, amballa2025can, yang2025diffusion}\endNoHyper{}.
For example, network planning~\NoHyper\cite{rappaport2002wireless}\endNoHyper{} requires coverage maps that evaluate signal strength across candidate transmitter and receiver placements.
Fingerprint-based localization~\NoHyper\cite{zhao2023nerf2, yang2025diffusion, yang2024orchloc}\endNoHyper{} depends on databases of~RF measurements collected across numerous transmitter positions at multiple receivers.
However, acquiring these measurements requires site surveys at thousands of transmitter--receiver positions, which is prohibitively labor-intensive and time-consuming~\cite{kar2018site, site_survey_cisco, shin2014mri, parralejo2021comparative}.

RF data synthesis reduces the need for exhaustive site surveys.
Classical ray tracing simulations~\cite{yun2015ray, egea2021opal, matlab_indoor_simulation, guo2024rad} trace propagation paths through Computer-Aided Design~(CAD) scene models.
However, accurate CAD models require precise scene geometry, material permittivity, and conductivity.
These properties are difficult to obtain in practice, making simulations infeasible.

Neural scene representations sidestep these requirements by learning~RF scene responses from measurements~\cite{zhao2023nerf2, zhang2024rf3dgs, wen2024wrfgs, yang2025gsrf, lu2024newrf}.
This enables synthesis of novel~RF data via differentiable rendering, without precise~CAD models or material priors.
Neural Radiance Field~(NeRF)-based~RF methods~\cite{zhao2023nerf2, lu2024newrf} pioneer this direction but incur prohibitive computational costs from dense volumetric sampling and multi-layer perceptron~(MLP) inference.
3D Gaussian Splatting~(3DGS)-based~RF methods~\cite{zhang2024rf3dgs, wen2024wrfgs, yang2025gsrf} replace volumetric sampling with~3D Gaussians for faster training and inference.

The rendering interface of~3DGS also shapes how these methods represent the two endpoints of an~RF link.
Optical~3DGS~\cite{kerbl20233d, yu2024mip, chen2024survey} is parameterized by a single variable, the camera viewpoint, since the scene is queried from one viewpoint at a time.
RF data synthesis instead depends on two variables, the transmitter position and the receiver position.
Existing~3DGS-based~RF methods expose the densely surveyed moving endpoint as an input and absorb the other, fixed endpoint into model weights.
In our data, the moving transmitters are surveyed across thousands of positions, while the receivers occupy a smaller set of fixed positions.
Without an input pathway for the fixed receiver position, these models learn receiver-specific radiance coefficients, as detailed in~\S\ref*{sec_preliminary}.

However, multi-receiver deployments in a scene are the norm in practice~\cite{site_survey_cisco, chataut2020massive, ajmal2024cell}.
For example, a~Bluetooth Low Energy~(BLE) localization system deploys dozens of gateways distributed across a facility~\cite{zhao2023nerf2}.
With a separate checkpoint for each receiver, the costs of a multi-receiver deployment grow in two ways.
\emph{\ding{182}~Scaling Cost}: synthesizing~RF data for~$N$ receivers within one scene requires training~$N$ independent models.
\emph{\ding{183}~Adding New Receivers.}
For per-receiver~3DGS methods, a new fixed receiver in the same scene requires a full transmitter survey and its own checkpoint.

We present~\ourSystem, a~3DGS-based~RF data synthesis framework that conditions one shared model on fixed receiver locations within a surveyed scene.
Our key insight is that scene geometry is receiver-independent while directional radiance is not: scatterers such as walls and furniture remain fixed regardless of receiver position, but the angles and distances to the receiver change.
This motivates a two-stage training procedure.
Stage~I learns shared~3D Gaussian geometry of the scene.
Stage~II freezes this geometry and jointly learns receiver-conditioned directional radiance.

When the receiver moves, both the global radiance pattern and each~Gaussian's individual contribution change.
To capture these effects, we propose two complementary conditioning modules.
A global module modulates harmonic coefficients based on receiver position, capturing the shared component of receiver-dependent variation.
A local module adjusts each~Gaussian based on its direction, distance, and occlusion relative to the receiver, capturing per-scatterer residuals.
Both modules are radiance-representation-agnostic and apply to existing~3DGS-based~RF methods, including spherical harmonics~(SH), Fourier--Legendre expansions~(FLE), and deformation networks~\cite{zhang2024rf3dgs, wen2024wrfgs, yang2025gsrf}.
For the~SH and~FLE backbones, a batched rasterizer uses the~Compute Unified Device Architecture~(CUDA) to render~$N$ receiver outputs in parallel while sharing projection and sorting.

\ourSystem introduces explicit fixed-receiver conditioning to~3DGS-based~RF synthesis, enabling predictions at receiver locations excluded from training within the surveyed region.
On~RSSI, it improves the~RF-3DGS~\NoHyper\cite{zhang2024rf3dgs}\endNoHyper{} and~GSRF~\NoHyper\cite{yang2025gsrf}\endNoHyper{} per-receiver baselines.
With the~WRF-GS+~\NoHyper\cite{wen2024wrfgs}\endNoHyper{} backbone, it eliminates two high-error outliers with a~$0.33$\,dBm increase in mean absolute error~(MAE).
On spatial spectrum and~CSI, it trails per-receiver baselines by~$0.8$ to~$1.7$\,dB.
Training is up to~$45\times$ faster, batched~FLE/SH inference up to~$7.6\times$ faster, and storage up to~$24\times$ smaller.

\vspace{\mylen}
\section{Preliminaries}\label{sec_preliminary}
\vspace{\mylen}

We review the~3DGS-based~RF data synthesis framework~\cite{zhang2024rf3dgs, wen2024wrfgs, yang2025gsrf} that~\ourSystem builds on.

\textbf{1)~Scene Representation.}
The~RF scene,~\eg~a conference room, is a collection of~3D Gaussian primitives~$\{\mathcal{G}_k\}_{k=1}^{K}$, each defined by four attributes:
\begin{equation}
    \mathcal{G}_k = (\mathbf{p}_k, \mathbf{C}_k, \phi_k, \tau_k).
\end{equation}
The mean~$\mathbf{p}_k \in \mathbb{R}^3$ and covariance~$\mathbf{C}_k \in \mathbb{R}^{3 \times 3}$ define the spatial geometry of each~Gaussian.
The radiance~$\phi_k$ represents the~RF signal emitted by the~$k$-th~Gaussian.
To model its directional dependency,~$\phi_k$ is parameterized by the coefficients of basis functions on the sphere, such as~SH~\cite{schonefeld2005spherical} or~FLE~\cite{cornelius2017spherical}.
Transmittance~$\tau_k$ models direction-independent attenuation through the~$k$-th~Gaussian.

\textbf{2)~Rendering.}
Given a receiver at position~$\mathbf{r}$ and a transmitter at position~$\mathbf{t}$,~RF data is synthesized as follows.
Rays are emitted from an emission surface centered at the transmitter~$\mathbf{t}$.
Each~Gaussian is projected onto this surface to identify ray--Gaussian intersections based on the ray direction and the~Gaussian mean~$\mathbf{p}_k$ and covariance~$\mathbf{C}_k$.
For each intersecting~Gaussian along a ray, the directional radiance~$\phi_k$ depends on the direction from the transmitter to the~Gaussian center:
\begin{equation}
\label{eq:dk}
    \hat{\mathbf{d}}_k = \left(\mathbf{p}_k - \mathbf{t}\right) \,/\, \|\mathbf{p}_k - \mathbf{t}\|_2.
\end{equation}
The received signal for each ray aggregates the radiance~$\phi_k$ and transmittance~$\tau_k$ contributions of all intersecting~Gaussians along that ray.
All~Gaussian attributes are then jointly optimized by minimizing a reconstruction loss between rendered and ground-truth~RF measurements.

\textbf{Receiver Dependence.}
In our setting,~RF data synthesis involves a densely surveyed moving endpoint and a fixed receiver endpoint.
In our datasets, the moving endpoint is the transmitter at~$\mathbf{t}$ and the fixed endpoint is the receiver at~$\mathbf{r}$.
Across prior formulations, the physical transmitter and receiver labels can be reversed.
With matched hardware, reciprocity aligns their fixed and moving endpoint roles, while directional measurements retain endpoint-specific conventions.
In Equation~\ref*{eq:dk}, the moving transmitter enters explicitly through~$\hat{\mathbf{d}}_k$, enabling synthesis at new transmitter positions.
The fixed receiver position~$\mathbf{r}$ enters implicitly through the ground-truth measurements used during training.
Consequently, each checkpoint represents the radiance field for its training receiver.
Additional receiver placements call for separate measurements and fitting.

\vspace{\mylen}
\section{Related Work}\label{sec_related_work}
\vspace{\mylen}

\textbf{NeRF-Based RF Data Synthesis.}
NeRF\textsuperscript{2}~\cite{zhao2023nerf2} pioneers neural~RF data synthesis by representing the scene as a continuous volumetric function parameterized by two large~MLPs for attenuation and radiance.
Several subsequent works study neural~RF synthesis~\NoHyper\cite{lu2024newrf, yang2025generalizable, yang2026corf, nerfapt2025, zeng2025voxelrf}\endNoHyper{}.
For example,~NeWRF~\cite{lu2024newrf} leverages direction-of-arrival~(DoA) priors to reduce the number of emitted rays, improving computational efficiency.
GRaF~\cite{yang2025generalizable} interpolates spatial spectra from geographically proximate transmitters, but requires reference spectra at the query receiver at inference, preventing zero-shot prediction at unseen receivers.
NeRF-based methods such as~NeRF\textsuperscript{2} represent receiver position through the ray origin and sampled directions, allowing queries at previously unseen receivers.
Dense volumetric sampling and repeated~MLP queries raise training and inference costs~\NoHyper\cite{mildenhall2021nerf, chen2021mvsnerf}\endNoHyper{}.

\textbf{3DGS-Based RF Data Synthesis.}
Existing~3DGS-based~RF methods~\cite{zhang2024rf3dgs, wen2024wrfgs, yang2025gsrf, swiftwrf2025, yang2025scalable} fall into two categories based on whether neural networks are adopted.
WRF-GS+~\cite{wen2024wrfgs} uses a deformation network to predict offsets for~Gaussian attributes from~Gaussian and transmitter positions, and~SwiftWRF~\cite{swiftwrf2025} follows a similar architecture with~2D Gaussians for faster rendering.
RF-3DGS~\cite{zhang2024rf3dgs} and~GSRF~\cite{yang2025gsrf} avoid neural networks.
RF-3DGS learns geometry from camera images, then retrains~SH coefficients from~RF measurements.
GSRF optimizes per-Gaussian~FLE radiance coefficients~\cite{cornelius2017spherical} as learnable complex-valued parameters.
These methods supply the moving endpoint as a spatial query and encode the fixed endpoint in a checkpoint, as described in~\S\ref*{sec_preliminary}.
A multi-receiver deployment therefore trains one checkpoint per fixed endpoint.
WRF-GS+ labels that fixed endpoint as the receiver, while~RF-3DGS and~GSRF label it as the transmitter.

\textbf{Radio Maps from Scene Layouts.}
RadioUNet~\NoHyper\cite{levie2021radiounet}\endNoHyper{} predicts scalar pathloss maps from~2D occupancy maps and transmitter positions using simulated data.
By contrast,~\ourSystem synthesizes~RSSI, spatial spectra, and~CSI from~RF measurements, including floorplan-free~BLE data.

\textbf{Relighting in Neural Rendering.}
Generalizing across receiver positions in~RF synthesis is structurally analogous to relighting in optical neural rendering, where both must generalize across a variable baked into the learned representation.
Several works address this by replacing the standard rendering equation with a physically-based rendering~(PBR) equation that takes lighting as an explicit input~\cite{boss2021nerd, zhang2021nerfactor, gao2024relightable, liang2024gs, saito2024relightable}.
For instance, Relightable 3DGS~\cite{gao2024relightable} augments each~Gaussian with surface normals, bidirectional reflectance distribution function~(BRDF) parameters, and incident light to compute outgoing radiance under arbitrary illumination.
Optical relighting uses~BRDF-based priors, whereas~RF multipath calls for a receiver-dependent signal model tailored to propagation.
\ourSystem therefore introduces lightweight learned conditioning modules to capture these receiver-dependent effects.
Its global module plays a role analogous to a learned environment map, while the global-only and additive variants in Table~\ref*{tab:ablation_ble} isolate the effects of the two conditioning choices.

\vspace{\mylen}
\section{Method}\label{sec_method}
\vspace{\mylen}

Consider a scene with transmitter positions~$\{\mathbf{t}_i\}_{i=1}^{M}$ and receiver positions~$\{\mathbf{r}_j\}_{j=1}^{N}$.
For each transmitter--receiver pair~$(\mathbf{t}_i, \mathbf{r}_j)$, the received~RF data~$S(\mathbf{t}_i, \mathbf{r}_j)$ is measured.
$S$ takes different forms across modalities: a scalar~RSSI value, a spatial spectrum matrix, or complex-valued~CSI.

\textbf{Prior Formulation~(Single-Receiver).}
Existing~3DGS-based~RF data synthesis methods learn a separate model~$\Theta_j$ for each receiver~$\mathbf{r}_j$:
\begin{equation}
    \hat{S}(\mathbf{t}, \mathbf{r}_j) = f(\mathbf{t};\, \Theta_j), \quad \forall\, \mathbf{t} \in \mathbb{R}^3.
\end{equation}
The per-receiver formulation requires~$N$ checkpoints for~$N$ receivers.
Extending the deployment to~$\mathbf{r}^* \notin \{\mathbf{r}_j\}$ requires receiver-specific measurements and model training.

\textbf{Our Formulation~(Receiver-Generalizable).}
Our goal is to learn a single model~$\Theta$ that generalizes across transmitter positions and fixed receiver positions within a surveyed, static scene:
\begin{equation}
    {\hat{S}(\mathbf{t}, \mathbf{r}) = f(\mathbf{t}, \mathbf{r};\, \Theta), \quad \mathbf{t} \in \mathcal{T}_{\mathrm{survey}},\, \mathbf{r} \in \mathcal{R}_{\mathrm{survey}},}
\end{equation}
Here~$\mathcal{T}_{\mathrm{survey}}$ and~$\mathcal{R}_{\mathrm{survey}}$ denote transmitter and candidate receiver locations in the surveyed deployment region, including receiver sites excluded from training.
The key challenge is to capture receiver-dependent variation in one parameter set while managing the accuracy--efficiency tradeoff.

\begin{figure*}[t]
\centering
{\includegraphics[width=\textwidth]{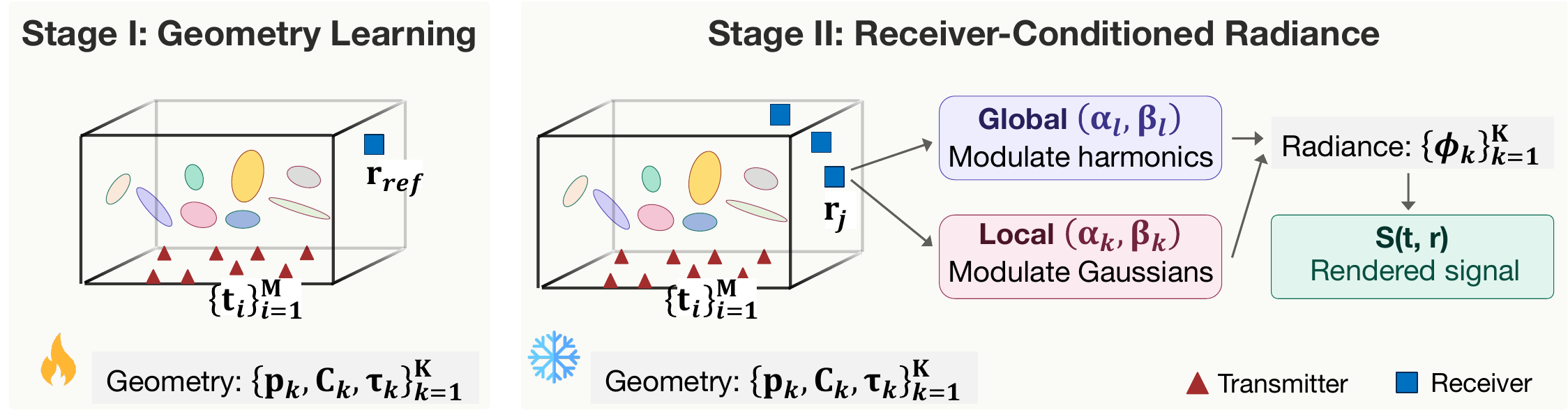}}
\vspace{-0.20in}
\caption{\ourSystem's two-stage architecture.~Stage~I learns receiver-independent~3D Gaussian geometry, and~Stage~II conditions radiance on the receiver via global and local modules in one unified model.}
\label{fig:overview}
\end{figure*}

\vspace{\mylen}
\subsection{Geometry-Radiance Decomposition Workflow}\label{sec_decomposition}
\vspace{\mylen}

The key insight behind~\ourSystem is that the~3D Gaussian attributes partition into two groups with distinct dependencies.
The geometric attributes, position~$\mathbf{p}_k$, covariance~$\mathbf{C}_k$, and transmittance~$\tau_k$, describe where scatterers are located, their spatial extent, and how they attenuate signals.
These are physical properties of the environment and are independent of the receiver position.
The radiance coefficients~$\phi_k$, in contrast, encode the directional signal response at each~Gaussian and change with receiver position.
Appendix~\ref{app_decomposition} provides a theoretical analysis showing that, under the rendering abstraction adopted by~3DGS-based~RF methods, receiver-dependent variation is naturally expressed through the radiance term while the geometric attributes remain receiver-independent.
Within this rendering abstraction, receiver-dependent effects are assigned to radiance while shared geometry supports all receivers.
This motivates the two-stage training procedure illustrated in Figure~\ref{fig:overview}.

\textbf{Stage~I: Geometry.}
All~Gaussian attributes~$(\mathbf{p}_k, \mathbf{C}_k, \phi_k, \tau_k)$ are jointly optimized on~RF measurements from a single reference receiver~$\mathbf{r}_{\mathrm{ref}}$ to learn the shared scene geometry.
After convergence, we freeze~$\{\mathbf{p}_k, \mathbf{C}_k, \tau_k\}_{k=1}^{K}$ for~Stage~II and discard the~Stage~I radiance coefficients.

\textbf{Stage~II: Receiver-Conditioned Radiance.}
With the geometry fixed, we introduce new radiance coefficients conditioned on the receiver position~$\mathbf{r}$:
\begin{equation}
    \phi_k(\hat{\mathbf{d}}_k, \mathbf{r}) = g(\phi_k^{\mathrm{base}}, \hat{\mathbf{d}}_k, \mathbf{r};\, \Theta_{\mathrm{cond}}),
\end{equation}
where~$\phi_k^{\mathrm{base}}$ are learnable base coefficients shared across all receivers, and~$g$ is a conditioning function that modulates them based on~$\mathbf{r}$.
$g$ consists of a global module that modulates directional harmonics and a local module that modulates individual~Gaussians, together parameterized by~$\Theta_{\mathrm{cond}}$.
Stage~II is trained on~RF measurements from all available receivers~$\{\mathbf{r}_j\}_{j=1}^{N}$ jointly, allowing the conditioning function to learn how receiver position affects directional radiance.

\vspace{\mylen}
\subsection{Stage~I: Geometry Learning}\label{sec_stage1}
\vspace{\mylen}

\textbf{Objective.}
We select an arbitrary reference receiver~$\mathbf{r}_{\mathrm{ref}}$ and optimize all~Gaussian attributes~$(\mathbf{p}_k, \mathbf{C}_k, \phi_k, \tau_k)$ jointly using~RF measurements from all transmitter positions.
The training objective minimizes the prediction error between the rendered and measured~RF data:
\begin{equation}\label{eq:stage1_loss}
    \{\mathbf{p}_k, \mathbf{C}_k, \phi_k, \tau_k\}_{k=1}^{K} 
    = \argmin_{\{\mathbf{p}_k, \mathbf{C}_k, \phi_k, \tau_k\}} 
    \frac{1}{M} \sum_{i=1}^{M} 
    \ell\!\left(\hat{S}(\mathbf{t}_i),\; S(\mathbf{t}_i, \mathbf{r}_{\mathrm{ref}})\right),
\end{equation}
where~$\ell$ is a task-specific loss determined by the~RF data modality:~$\ell_1$ for scalar~RSSI,~$\ell_2$ for complex-valued~CSI, or a composite loss for spatial spectrum data.
The rendered signal~$\hat{S}(\mathbf{t}_i)$ is computed via the rendering equation in Equation~\eqref{eq:render} of Appendix~\ref{app_decomposition}.
Because Gaussian geometry is receiver-independent, this objective learns shared attributes from measurements at the reference receiver.

\textbf{Representation Agnosticism.}
Stage~I is compatible with any directional radiance representation, including spherical harmonics~\cite{schonefeld2005spherical},~MLP-based radiance networks~\cite{wen2024wrfgs}, or~Fourier--Legendre expansions~\cite{cornelius2017spherical}, since the decomposition is independent of how~$\phi_k$ is parameterized.
After convergence, we freeze~$\{\mathbf{p}_k, \mathbf{C}_k, \tau_k\}_{k=1}^{K}$ and discard the~Stage~I radiance coefficients.

\vspace{\mylen}
\subsection{Stage~II: Receiver-Conditioned Radiance}\label{sec_conditioning}
\vspace{\mylen}

\textbf{Objective.}
With the geometry frozen from~Stage~I, we introduce a conditioning function~$g_{\Theta}$ that produces receiver-specific radiance from a shared set of learnable base coefficients~$\{\phi_k^{\mathrm{base}}\}_{k=1}^{K}$:
\begin{equation}\label{eq:stage2_radiance}
    \phi_k(\hat{\mathbf{d}}_k, \mathbf{r}) 
    = g_{\Theta}(\phi_k^{\mathrm{base}}, \hat{\mathbf{d}}_k, \mathbf{r}).
\end{equation}
The training objective optimizes~$\{\phi_k^{\mathrm{base}}\}_{k=1}^{K}$ and~$\Theta$ jointly across all available receivers:
\begin{equation}\label{eq:stage2_loss}
    \{\phi_k^{\mathrm{base}},\, \Theta\}
    = \argmin_{\{\phi_k^{\mathrm{base}},\, \Theta\}}
    \frac{1}{MN} \sum_{j=1}^{N} \sum_{i=1}^{M}
    \ell\!\left(\hat{S}(\mathbf{t}_i, \mathbf{r}_j),\; S(\mathbf{t}_i, \mathbf{r}_j)\right).
\end{equation}
The key question is how to design~$g_\Theta$.
A naive approach would condition on the raw receiver position~$\mathbf{r}$ alone, but the decomposition reveals richer structure: the receiver-dependence of~$\phi_k$ arises specifically from the last-segment propagation from each scatterer to the receiver.
We formalize this below and derive three constraints that dictate the architecture of~$g_{\Theta}$: the modulation should be affine, it should vary per scatterer, and it can share weights across scatterers.

\vspace{\mylen}
\subsubsection{Design Constraints from the Propagation Factor}
\vspace{\mylen}
\begin{proposition}[Last-Segment Factorization]\label{prop:factorization}
Under the multipath model in Definition~\ref*{def:multipath} of Appendix~\ref*{app_decomposition}, the radiance of~Gaussian~$k$ at position~$\mathbf{p}_k$ factorizes into a receiver-independent base~$\phi_k^{\mathrm{base}}$ and a receiver-dependent last-segment propagation factor~$\eta_{\mathrm{eff}}$:
\begin{equation}\label{eq:factorization}
    \phi_k(\hat{\mathbf{d}}_k, \mathbf{r}) 
    = \eta_{\mathrm{eff}}(\mathbf{p}_k, \mathbf{r}) 
    \cdot \phi_k^{\mathrm{base}}(\hat{\mathbf{d}}_k, \mathbf{t}),
\end{equation}
where~$\phi_k^{\mathrm{base}}$ aggregates all path contributions from the transmitter to scatterer~$k$, and:
\begin{equation}\label{eq:eta_eff}
    \eta_{\mathrm{eff}}(\mathbf{p}_k, \mathbf{r}) 
    = \underbrace{\frac{\lambda}{4\pi\|\mathbf{p}_k - \mathbf{r}\|}
    \cdot e^{-j\frac{2\pi}{\lambda}\|\mathbf{p}_k - \mathbf{r}\|}}_{\text{free-space: path loss $+$ phase}}
    \cdot \underbrace{\prod_{m \in \mathcal{I}(k, \mathbf{r})} \tau_m}_{\text{occlusion}}
    = \frac{\lambda}{4\pi d_k} \cdot e^{-j \frac{2\pi}{\lambda} d_k} \cdot T_k.
\end{equation}
where~$d_k = \|\mathbf{p}_k - \mathbf{r}\|$ is the last-segment distance,~$T_k = \prod_{m \in \mathcal{I}(k, \mathbf{r})} \tau_m$ is the occlusion product, and~$\mathcal{I}(k, \mathbf{r})$ is the set of scatterers intersected by the segment from~$\mathbf{p}_k$ to~$\mathbf{r}$.
The proof is in Appendix~\ref{app_factorization}.
\end{proposition}

This factorization has three corollaries that dictate the conditioning architecture, each following from the structure of~$\eta_{\mathrm{eff}}$, with proofs in Appendix~\ref{app_corollary_proofs}.

\begin{corollary}[Global-local decomposition]\label{cor:global_local}
Decomposing the propagation factor in Proposition~\ref{prop:factorization} as~$\eta_{\mathrm{eff}}(\mathbf{p}_k, \mathbf{r}) = \bar{\eta}(\mathbf{r}) + \delta_k(\mathbf{r})$, where~$\bar{\eta}(\mathbf{r}) = \frac{1}{K}\sum_k \eta_{\mathrm{eff}}(\mathbf{p}_k, \mathbf{r})$ is the scatterer-averaged effect and~$\delta_k(\mathbf{r})$ is the per-scatterer residual, Equation~\eqref{eq:factorization} becomes:
\begin{equation}\label{eq:anova}
    \phi_k(\hat{\mathbf{d}}_k, \mathbf{r}) = \underbrace{\bar{\eta}(\mathbf{r}) \cdot \phi_k^{\mathrm{base}}}_{\text{global (shared)}}
    \;+\; \underbrace{\delta_k(\mathbf{r}) \cdot \phi_k^{\mathrm{base}}}_{\text{local (per-}k\text{)}}.
\end{equation}
Since~$\delta_k \neq 0$ for almost all~$\mathbf{r}$, both branches are required.
\end{corollary}

\begin{corollary}[Affine modulation suffices]\label{cor:affine}
The complex-valued product in Equation~\eqref{eq:anova} can be represented by element-wise affine modulation~$\left(1 + \boldsymbol{\alpha}\right) \odot \boldsymbol{\phi}^{\mathrm{base}}_k + \boldsymbol{\beta}$ with~$\boldsymbol{\alpha}, \boldsymbol{\beta} \in \mathbb{C}^d$, which jointly parameterizes the amplitude scaling and phase rotation induced by~$\eta_{\mathrm{eff}}$.
\end{corollary}

\begin{corollary}[Index-free form~$\Rightarrow$ shared weights]\label{cor:shared}
$\eta_{\mathrm{eff}}$ depends on~$k$ only through the inputs~$(d_k, T_k)$, which permits a single shared network.
\end{corollary}

\vspace{\mylen}
\subsubsection{Conditioning Architecture}
\vspace{\mylen}

We now instantiate the corollaries as architectural components.
Given a receiver position~$\mathbf{r}$, the conditioning function~$g_{\Theta}$ proceeds in three steps:
(i)~$\mathbf{r}$ is mapped to a feature~$\gamma(\mathbf{r})$ via a learnable~Fourier encoding that captures the two-scale structure of~$\eta_{\mathrm{eff}}$.
(ii)~A global branch applies affine modulation to all~Gaussians based on~$\gamma(\mathbf{r})$, as required by Corollaries~\ref*{cor:global_local} and~\ref*{cor:affine}.
(iii)~A local branch refines each~Gaussian individually using per-scatterer geometric and propagation features, as required by Corollary~\ref*{cor:global_local}, with shared weights across all~Gaussians by Corollary~\ref*{cor:shared}.

\textbf{Receiver Encoding.}
The propagation factor~$\eta_{\mathrm{eff}}$ in Equation~\eqref{eq:eta_eff} varies at two spatial scales: the path loss~$\frac{1}{\|\mathbf{p}_k - \mathbf{r}\|}$ changes over room-scale distances, while the phase~$e^{-j\frac{2\pi}{\lambda}\|\mathbf{p}_k - \mathbf{r}\|}$ oscillates at the centimeter-scale wavelength~$\lambda$.
We use a learnable~Fourier encoding of~$\mathbf{r}$ for both scales:
\begin{equation}\label{eq:rx_enc}
    \gamma(\mathbf{r}) = \bigl[\sin(\omega_1^\top \mathbf{r}),\, 
    \cos(\omega_1^\top \mathbf{r}),\, \ldots,\, 
    \sin(\omega_F^\top \mathbf{r}),\, \cos(\omega_F^\top \mathbf{r})\bigr],
\end{equation}
where the frequency vectors~$\{\omega_f\}_{f=1}^{F}$ are initialized at logarithmically spaced scales to span both the room-scale path loss and the wavelength-scale phase, and are learned end-to-end so that they adapt to the scene-specific ratio between room geometry and operating wavelength.

\textbf{Global Branch.}
The global branch captures shared receiver-dependence across all scatterers.
When the receiver moves, the relative importance of directional harmonics shifts.
Low-order components capture broad propagation patterns and are less sensitive to receiver position, while high-order components encode fine-grained multipath structure varying sharply with~$\mathbf{r}$.
To model this, an~MLP generates per-component affine parameters from the receiver encoding and component identity:
\begin{equation}\label{eq:film}
    (\boldsymbol{\alpha}_\ell,\, \boldsymbol{\beta}_\ell) 
    = \mathrm{MLP}_{\mathrm{global}}\!\left([\gamma(\mathbf{r});\, 
    \mathbf{c}_\ell;\, \mathbf{e}_\ell]\right),
\end{equation}
where~$\mathbf{c}_\ell$ encodes the structural identity of radiance component~$\ell$,~\ie~its degree and order in the~Fourier--Legendre basis, and~$\mathbf{e}_\ell$ is a learnable embedding.
Since the radiance is complex-valued,~$\boldsymbol{\alpha}_\ell$ and~$\boldsymbol{\beta}_\ell$ operate on both real and imaginary parts, jointly modulating amplitude and phase as required by Corollary~\ref{cor:affine}.
The same modulation applies to all~Gaussians:
\begin{equation}\label{eq:film_apply}
    \tilde{\phi}_k^{(\ell)}(\mathbf{r}) = (1 + \boldsymbol{\alpha}_\ell) 
    \odot \phi_k^{\mathrm{base},(\ell)} + \boldsymbol{\beta}_\ell, 
    \quad \forall\, k = 1, \ldots, K.
\end{equation}
The total cost is~$\mathcal{O}(L)$ independent of~Gaussian count, where~$L$ is the number of radiance components, keeping the global branch lightweight even for scenes with many~Gaussians.

\textbf{Local Branch.}
The local branch captures per-scatterer variation that the global branch cannot represent.
In~RF propagation, a scatterer's contribution to the received signal depends critically on its individual relationship to the receiver.
A wall in direct line-of-sight dominates the signal, while a scatterer behind multiple obstacles contributes little regardless of the global harmonic adjustment.
To model this, we construct a six-dimensional feature for each~Gaussian~$k$ that directly mirrors the three quantities governing~$\eta_{\mathrm{eff}}$ in Equation~\eqref{eq:eta_eff}:~\emph{direction},~\emph{distance}, and~\emph{occlusion}.

The~\emph{direction}~$\hat{\mathbf{v}}_k = \frac{\mathbf{r} - \mathbf{p}_k}{\|\mathbf{r} - \mathbf{p}_k\|}$ and~\emph{distance}~$d_k = \|\mathbf{r} - \mathbf{p}_k\|$ are computed directly from the frozen~Gaussian positions and the receiver position, capturing the free-space path loss~$\frac{1}{d_k}$ and phase rotation~$e^{-j\frac{2\pi}{\lambda} d_k}$ on the last segment.
These features characterize the free-space component of~$\eta_{\mathrm{eff}}$ per scatterer without additional learned geometry, since the~Gaussian positions are frozen from~Stage~I.

The~\emph{occlusion}, however, requires the cumulative scene density along the path from~$\mathbf{p}_k$ to~$\mathbf{r}$, which we obtain from the frozen~Gaussian geometry without additional supervision.
At the start of~Stage~II, we voxelize the scene by splatting each~Gaussian's transmittance into a~3D occupancy grid~$\mathcal{V}$ based on its learned position~$\mathbf{p}_k$ and covariance~$\mathbf{C}_k$.
For each~Gaussian~$k$, we cast a ray from~$\mathbf{p}_k$ to~$\mathbf{r}$ and probe the grid at~$S$ evenly spaced points~$\{\mathbf{q}_{k,s}\}_{s=1}^{S}$:
\begin{equation}\label{eq:transmittance}
    T_k = \prod_{s=1}^{S} \bigl(1 - \mathcal{V}(\mathbf{q}_{k,s})\bigr), \qquad
    \bar{\rho}_k = \frac{1}{S}\sum_{s=1}^{S} \mathcal{V}(\mathbf{q}_{k,s}),
\end{equation}
where~$\mathcal{V}(\cdot)$ denotes the trilinearly interpolated occupancy.
The transmittance~$T_k$ approximates the occlusion product~$\prod_{m \in \mathcal{I}(k, \mathbf{r})} \tau_m$ from Equation~\eqref{eq:eta_eff}:~$T_k \to 1$ indicates a clear last segment, while~$T_k \to 0$ indicates heavy occlusion.
The mean density~$\bar{\rho}_k$ provides complementary information, distinguishing a single concrete wall where~$T_k$ is low but~$\bar{\rho}_k$ remains moderate, from distributed clutter where both~$T_k$ and~$\bar{\rho}_k$ are high.
A single shared~MLP processes all~Gaussians:
\begin{equation}\label{eq:local}
    (\boldsymbol{\alpha}_k^{\mathrm{loc}},\, \boldsymbol{\beta}_k^{\mathrm{loc}}) 
    = \mathrm{MLP}_{\mathrm{local}}\!\left(\left[\hat{\mathbf{v}}_k;\, d_k;\, 
    T_k;\, \bar{\rho}_k\right]\right).
\end{equation}

\textbf{Combined Modulation.}
The two branches compose sequentially: the global branch produces~$\tilde{\phi}_k^{(\ell)}(\mathbf{r})$ via Equation~\eqref{eq:film_apply}, and the local branch then applies:
\begin{equation}\label{eq:combined_local}
    \phi_k(\hat{\mathbf{d}}_k, \mathbf{r}) = (1 + \boldsymbol{\alpha}_k^{\mathrm{loc}}) 
    \odot \tilde{\phi}_k(\mathbf{r}) + \boldsymbol{\beta}_k^{\mathrm{loc}},
\end{equation}
mirroring the factored structure of~$\eta_{\mathrm{eff}}$ in Equation~\eqref{eq:eta_eff}.
Both branches are zero-initialized.
Consequently, the conditioning begins as an identity transformation and supports stable training from the~Stage~I solution.
The total cost is~$\mathcal{O}(L + K)$, and the occupancy grid is built once with no gradient overhead.
At inference,~$g_{\Theta}$ accepts receiver locations~$\mathbf{r}^*$ excluded from training within the surveyed scene.

\vspace{\mylen}
\subsection{Single-Kernel Multi-Receiver Rasterization}
\vspace{\mylen}

Since the~Stage~I geometry is receiver-independent, per-Gaussian projection, depth sorting, and basis evaluation at~$\hat{\mathbf{d}}_k$ depend only on the transmitter and are shared across all~$N$ receivers.
For the~FLE and~SH backbones, receiver-specific radiance~$\phi_k(\hat{\mathbf{d}}_k, \mathbf{r}_j)$ is evaluated in a batched operation after shared geometry work.
Their rasterizer launches one~CUDA kernel that tiles the rendering sphere along elevation~$\theta$ and azimuth~$\phi$, with~$\mathrm{gridDim} = (\mathrm{tiles}_\theta, \mathrm{tiles}_\phi, N)$.
The receiver index maps to the third grid dimension, with each block reading its signal slice and writing its output slot.
The shared projection and sorting work is amortized across receiver queries.
The third grid dimension scales with~$N$ to evaluate receiver-specific radiance and write~$N$ outputs, as detailed in Appendix~\ref*{app_multirx_impl}.

\vspace{\mylen}
\section{Experiments}\label{sec_evaluation}
\vspace{\mylen}

\textbf{Implementation.}
We implement \ourSystem{} in PyTorch with custom CUDA kernels for~RF signal tracing and rasterization.
All experiments are conducted on a single~NVIDIA H100~GPU.
Full implementation and hyperparameter details are in Appendix~\ref{app_implementation}.

\textbf{Datasets.}
We evaluate \ourSystem{} on three~RF datasets:
\ding{182}~\textit{BLE RSSI}~\cite{zhao2023nerf2},~\ding{183}~\textit{RF Identification~(RFID) spatial spectrum}~(angular power distribution at the receiver) simulated via~Sionna~\cite{hoydis2022sionna, hoydis2023sionna}, and~\ding{184}~\textit{WiFi CSI}~\cite{mamimo}.
The datasets span scalar~RSSI~(21~receivers), spectrum matrices~(21~receivers), and complex-valued~CSI~(8~receivers).
Full dataset details are provided in Appendix~\ref{app_dataset}.

\textbf{Metrics.}
We adopt metrics matched to each dataset's signal modality.
\ding{182}~\textit{Mean absolute error~(MAE)} in~dBm for~BLE~RSSI.
\ding{183}~\textit{Mean squared error~(MSE)},~\textit{peak signal-to-noise ratio~(PSNR)}, and~\textit{structural similarity index~(SSIM)}~\cite{horeimage} for~RFID spatial spectrum.
\ding{184}~\textit{Signal-to-noise ratio~(SNR)} in~dB for~WiFi~CSI, computed across~$N_{\mathrm{sc}}=26$ subcarriers:
\begin{equation}
\text{SNR} = -10 \log_{10} \bigl( \textstyle\sum_{k} | \hat{h}_k - h_k |^2 \bigr/ \textstyle\sum_{k} | h_k |^2 \bigr)
\end{equation}
where~$h_k$ and~$\hat{h}_k$ are the ground-truth and predicted complex~CSI at subcarrier~$k$.
For each metric, we first compute the per-receiver mean by averaging over all test samples.
We then report the mean and standard deviation across these per-receiver means.
This convention emphasizes consistency across receivers rather than per-sample prediction variance.

\textbf{Baselines.}
We compare~\ourSystem{} against five baselines.
\ding{182}~\textit{NeRF$^2$}~\cite{zhao2023nerf2} is an~MLP-based neural~RF radiance field, representing the leading~NeRF-based approach.
\ding{183}~\textit{GRaF}~\cite{yang2025generalizable} is a~NeRF-based cross-scene spectrum synthesis method that interpolates reference spectra from neighboring transmitters at the query receiver.
\ding{184}~\textit{GSRF}~\cite{yang2025gsrf} is a complex-valued~FLE-based~3DGS method for~RF synthesis, representing the non-neural~3DGS approach, trained per receiver.
\ding{185}~\textit{RF-3DGS}~\cite{zhang2024rf3dgs} is the~SH-based~3DGS approach for~RF synthesis, also trained per receiver.
Since camera-based geometry initialization is unavailable in RF-only settings, we adopt the geometry-learning pipeline of~GSRF~\cite{yang2025gsrf} and replace~FLE with~SH to match the~RF-3DGS radiance formulation.
\ding{186}~\textit{WRF-GS+}~\cite{wen2024wrfgs} adapts deformable~3DGS from the optical domain to~RF synthesis by conditioning the deformation network on transmitter position, representing the leading neural~3DGS approach, also trained per receiver.

\begin{table*}[t]
\centering
\caption{Overall~RF data synthesis performance.
\ourSystem{} uses a single shared model for all~$N$ receivers.
Note:~GRaF reports spectrum results only, as its image-format backbone does not transfer to scalar~RSSI or complex~CSI without redesigning its entire architecture.}
\label{tab:overall}
\resizebox{\linewidth}{!}{%
\begin{tabular}{lllllll}
\toprule
& & \textbf{BLE RSSI} & \multicolumn{3}{c}{\textbf{RFID Spectrum}} & \textbf{WiFi CSI} \\
\cmidrule(lr){3-3} \cmidrule(lr){4-6} \cmidrule(lr){7-7}
Method & \#Models & MAE$\downarrow$ & MSE$\downarrow$ & PSNR$\uparrow$ & SSIM$\uparrow$ & SNR$\uparrow$ \\
\midrule
NeRF$^2$~\cite{zhao2023nerf2}         & 1   & $3.40_{\pm0.14}$ & $0.0040_{\pm0.0005}$ & $24.58_{\pm0.51}$ & $0.868_{\pm0.011}$ & $21.00_{\pm0.46}$ \\
GRaF~\cite{yang2025generalizable}         & 1   & --- & $0.0029_{\pm0.0004}$ & $26.40_{\pm0.58}$ & $0.903_{\pm0.007}$ & --- \\
\midrule
RF-3DGS~\cite{zhang2024rf3dgs}        & $N$ & $4.09_{\pm0.66}$ & $0.0054_{\pm0.0004}$ & $23.11_{\pm0.32}$ & $0.801_{\pm0.008}$ & $21.89_{\pm0.56}$ \\
\rowcolor{black!8} \withours{RF-3DGS}                        & 1   & $3.94_{\pm0.48}$ & $0.0068_{\pm0.0007}$ & $22.22_{\pm0.38}$ & $0.791_{\pm0.012}$ & $20.19_{\pm0.39}$ \\
\midrule
GSRF~\cite{yang2025gsrf}              & $N$ & $3.74_{\pm0.41}$ & $0.0052_{\pm0.0004}$ & $23.29_{\pm0.35}$ & $0.831_{\pm0.010}$ & $21.91_{\pm0.52}$ \\
\rowcolor{black!8} \withours{GSRF}                           & 1   & $3.57_{\pm0.24}$ & $0.0065_{\pm0.0006}$ & $22.47_{\pm0.35}$ & $0.780_{\pm0.012}$ & $21.12_{\pm0.49}$ \\
\midrule
WRF-GS+~\cite{wen2024wrfgs}           & $N$ & $5.84_{\pm5.85}$ & $0.0034_{\pm0.0005}$ & $25.81_{\pm0.66}$ & $0.892_{\pm0.012}$ & $21.84_{\pm0.52}$ \\
\rowcolor{black!8} \withours{WRF-GS+}                       & 1   & $6.17_{\pm0.71}$ & $0.0048_{\pm0.0007}$ & $24.27_{\pm0.67}$ & $0.860_{\pm0.011}$ & $20.55_{\pm0.38}$ \\
\bottomrule
\end{tabular}
}
\end{table*}

\vspace{\mylen}
\subsection{Overall Performance Against Per-Receiver Baselines}\label{sec_overall}
\vspace{\mylen}

We evaluate~\ourSystem against per-receiver baselines on two axes, synthesis quality in Table~\ref{tab:overall} and computational cost in Table~\ref{tab:cost}.
Appendices~\ref*{app_per_receiver}--\ref*{app_capacity} give receiver-level results and analyses.

\textbf{Synthesis Quality.}
Table~\ref*{tab:overall} shows the cost--accuracy tradeoff of one shared receiver-conditioned model: mean~BLE~RSSI error falls for two backbones and rises by~$0.33$\,dBm for~WRF-GS+, while spectrum and~CSI trail per-receiver baselines by~$0.8$ to~$1.7$\,dB.
NeRF-based approaches achieve the highest spectrum quality, with~GRaF~\cite{yang2025generalizable} reaching~$26.40$\,dB~PSNR and~{NeRF}$^2$ reaching~$24.58$\,dB, but at prohibitive cost~\cite{zhang2024rf3dgs, wen2024wrfgs, yang2025gsrf}.
\ourSystem instead retains an explicit~3DGS representation to reduce training and inference costs.

On~BLE~RSSI,~\ourSystem uses one shared model across receivers for each backbone.
\withours{GSRF} improves~MAE from~$3.74$ to~$3.57$\,dBm, and~\withours{RF-3DGS} from~$4.09$ to~$3.94$\,dBm.
The effect is most pronounced for~WRF-GS+, whose per-receiver baseline has two high-error receivers above~$18$\,dBm.
\withours{WRF-GS+} removes those outliers, while mean~MAE increases by~$0.33$\,dBm.
Appendix~\ref*{app_per_receiver} gives the per-receiver breakdown.

\begin{wraptable}{r}{0.43\linewidth}
\centering
\caption{Training, inference, and storage for serving~$N = 21$ receivers on~RFID.}
\label{tab:cost}
\resizebox{\linewidth}{!}{%
\begin{tabular}{l l l l}
\toprule
Method & Train~(h) & Infer.~(ms) & Stor.~(MB) \\
\midrule
NeRF$^2$~\cite{zhao2023nerf2}          & $14.3$ & $2436.4$ & $2.6$ \\
GRaF~\cite{yang2025generalizable}      & $16.2$ & $59891.4$ & $143.7$ \\
\midrule
RF-3DGS~\cite{zhang2024rf3dgs}         & $2.8$ & $156.4$ & $181.9$ \\
\rowcolor{black!8} \withours{RF-3DGS}                     & $0.4$ & $35.3$ & $8.8$ \\
\midrule
GSRF~\cite{yang2025gsrf}               & $5.5$ & $238.1$ & $548.0$ \\
\rowcolor{black!8} \withours{GSRF}                        & $0.6$ & $31.2$ & $26.1$ \\
\midrule
WRF-GS+~\cite{wen2024wrfgs}            & $13.6$ & $52.4$ & $75.6$ \\
\rowcolor{black!8} \withours{WRF-GS+}                     & $0.3$ & $42.5$ & $3.2$ \\
\bottomrule
\end{tabular}
}
\end{wraptable}

On~RFID spectrum and~WiFi~CSI,~\ourSystem trails per-receiver baselines by~$0.8$ to~$1.7$\,dB.
On spatial spectrum, the~$90\!\times\!360$ angular signal has fine directional structure, and the shared variants trail their corresponding per-receiver~3DGS baselines by~$0.8$ to~$1.5$\,dB~PSNR.
On~CSI, the complex-valued~FLE backbone~\withours{GSRF} has the smallest gap among the evaluated~3DGS variants, dropping~$0.79$\,dB from~$21.91$ to~$21.12$.
Appendix~\ref{app_csi_vis} provides qualitative~CSI predictions tracking ground truth in amplitude and phase.
The shared model also supports queries at fixed receivers excluded from training within the surveyed scene.
These results hold across three architecturally distinct~3DGS backbones,~\ie~RF-3DGS~\cite{zhang2024rf3dgs},~GSRF~\cite{yang2025gsrf}, and~WRF-GS+~\cite{wen2024wrfgs}.
The three backbones use real~SH, complex-valued~FLE, and deformation networks, respectively.
Our conditioning modules are compatible with each in the evaluated datasets.

\textbf{Computational Cost.}
Table~\ref*{tab:cost} summarizes end-to-end~RFID costs at~$N=21$, with per-backbone details in Appendix~\ref*{app_computation_cost}.
\ourSystem accelerates training by~$7\times$ to~$45\times$ over per-receiver baselines by reducing~$N$ independent training runs to a single shared-model fit.
Our multi-receiver rasterizer renders all~$N$ receivers in one batched~CUDA call.
At~$N=21$, per-transmitter inference is~$7.6\times$ faster than per-receiver~GSRF,~$4.4\times$ faster than~RF-3DGS, and~$1.2\times$ faster than~WRF-GS+.
Compared with~NeRF$^2$ and~GRaF, the shared~3DGS variants are~$60\times$ to~$1900\times$ faster.
Replacing~$N$ per-receiver checkpoints with one shared model compresses the deployable footprint by roughly~$21\times$.
The~FLE and~SH variants share projection and sorting across receiver queries.
WRF-GS+ retains a per-receiver render loop and has a smaller inference gain.

\vspace{\mylen}
\subsection{Generalization to Unseen Receivers}\label{sec_generalization}
\vspace{\mylen}

\begin{wrapfigure}{r}{0.65\linewidth}
\centering
\subfigure[BLE~RSSI]{
\label{fig_vis_a}
\includegraphics[width=.48\linewidth]{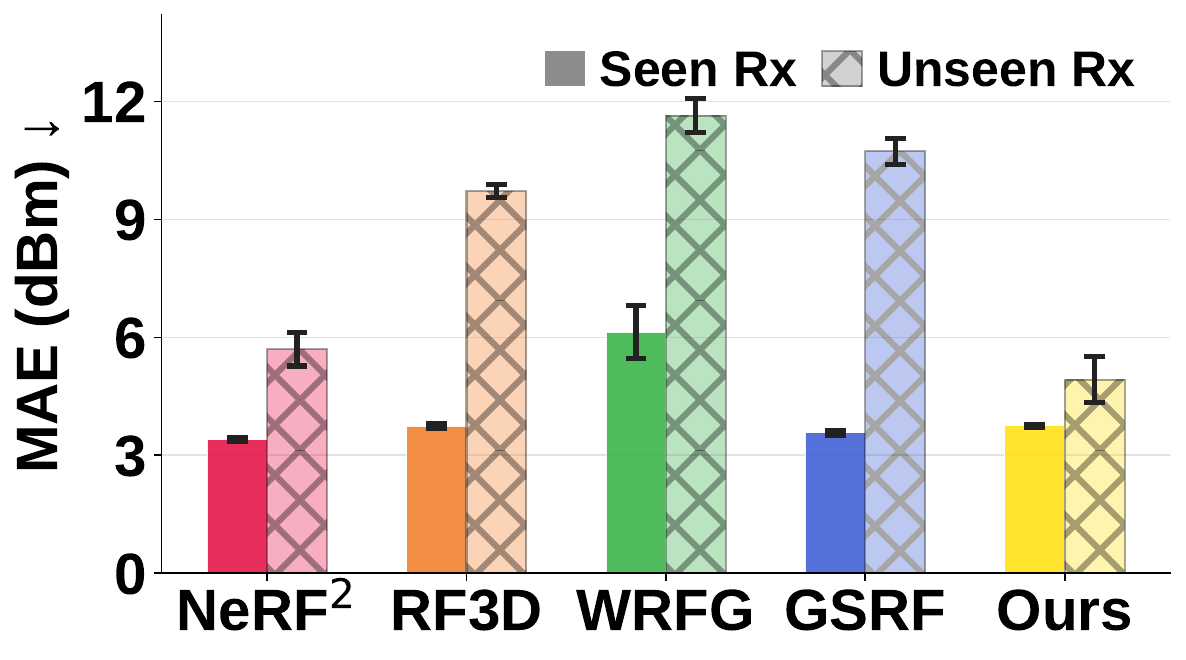}}
\subfigure[RFID~Spectrum]{
\label{fig_vis_b}
\includegraphics[width=.48\linewidth]{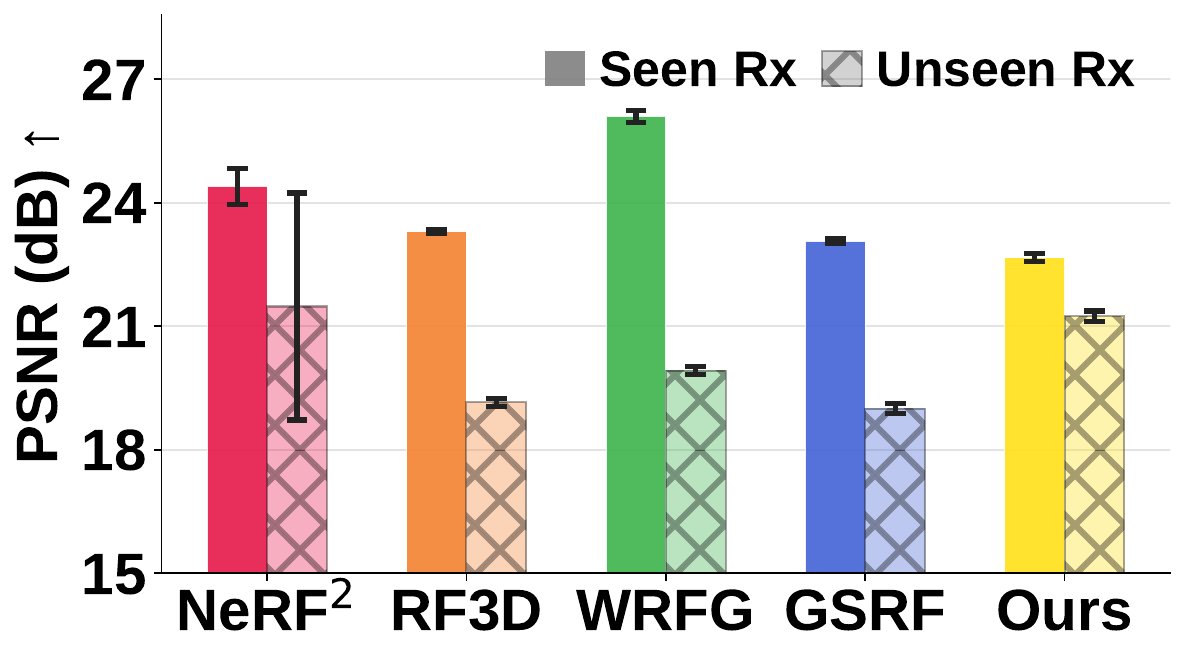}}
\caption{Unseen-receiver generalization.
``Ours'' is~\withours{GSRF}, and~RF3D\,/\,WRFG denote~RF-3DGS~\cite{zhang2024rf3dgs}\,/\,WRF-GS+~\cite{wen2024wrfgs}.}
\label{fig_unseen_rx}
\end{wrapfigure}

We partition the~$21$ receivers into~$3$ folds of~$7$, train on~$14$, and evaluate on all~$21$ receivers.
We exclude~WiFi~CSI because its~$8$ receivers are too few for meaningful seen/unseen splits.
Per-receiver baselines use the nearest trained receiver's checkpoint at unseen positions, while receiver-conditioned methods are queried directly.
{GRaF}~\cite{yang2025generalizable} is excluded because it requires reference spectra from the query receiver.

Figure~\ref{fig_unseen_rx} reports the mean and standard deviation across folds.
``Ours'' denotes~\withours{GSRF}, the variant closest to its per-receiver baseline in Table~\ref{tab:overall}.

For~BLE~RSSI, per-receiver~GSRF,~RF-3DGS, and~WRF-GS+ rise from~$3.6$--$6.1$\,dBm seen~MAE to~$9.7$--$11.6$\,dBm at unseen receivers.
On~RFID spectrum, they lose~$4$--$6$\,dB~PSNR.
WRF-GS+ falls from~$26.1$ to~$19.9$\,dB.
Appendix~\ref*{app_spectrum_vis} shows examples.

\textbf{Receiver Generalization.}
{NeRF}$^2$, which conditions a continuous neural field on the receiver position~$\mathbf{r}$, degrades much more gracefully, moving from~$3.40$ to~$5.70$\,dBm on~BLE.
\ourSystem achieves the lowest unseen~MAE of~$4.92$\,dBm on~BLE while matching the best seen performance at~$3.75$\,dBm.
On spectrum,~\ourSystem nearly matches~{NeRF}$^2$ on unseen~PSNR at~$21.25$ versus~$21.48$\,dB, while staying within~$2$\,dB of its own seen performance on both datasets.
Appendix~\ref{app_spectrum_vis} visualizes the predicted spectra:~\ourSystem recovers sharp directional peaks at unseen receivers,~{NeRF}$^2$ predictions are overly smooth, and~WRF-GS+ collapses to the closest seen-receiver pattern.
Receiver conditioning supports prediction at unseen fixed receiver positions within the surveyed scene.

\begin{wraptable}{r}{0.4\linewidth}
\centering
\caption{Ablation study on~BLE~RSSI.}
\label{tab:ablation_ble}
\small
\resizebox{\linewidth}{!}{%
\begin{tabular}{l c}
\toprule
Variant & MAE$\downarrow$ (dBm) \\
\midrule
\textbf{Full~\ourSystem{} -- \withours{GSRF}}        & $\mathbf{3.57_{\pm 0.24}}$ \\
\midrule
Joint training (no two-stage)   & $4.44_{\pm 0.39}$ \\
\midrule
Global branch only              & $4.11_{\pm 0.40}$ \\
Local branch only               & $3.91_{\pm 0.28}$ \\
\midrule
Additive modulation             & $4.02_{\pm 0.36}$ \\
\midrule
Without occlusion features      & $3.95_{\pm 0.25}$ \\
\bottomrule
\end{tabular}
}
\end{wraptable}

\vspace{\mylen}
\subsection{Ablation Study}
\vspace{\mylen}

We ablate each architectural design on~BLE~RSSI with~\withours{GSRF}, training and evaluating on all~$21$ receivers.
\ding{182}~Joint training raises~MAE from~$3.57$ to~$4.44$\,dBm, the largest gap in the table, 
consistent with the decomposition derived in Appendix~\ref{app_decomposition}:
jointly optimizing geometry and radiance prevents geometry from converging to a receiver-independent solution.
\ding{183}~Removing the local or global branch raises~MAE to~$4.11$ and~$3.91$\,dBm, respectively, showing neither alone suffices as predicted by Corollary~\ref{cor:global_local}: the global branch captures the scatterer-averaged effect but cannot represent per-scatterer residuals, while the local branch captures residuals but cannot efficiently model the shared component.
\ding{184}~Replacing complex affine modulation~$\left(1 + \boldsymbol{\alpha}\right) \odot \boldsymbol{\phi}^{\mathrm{base}} + \boldsymbol{\beta}$ with an additive shift~$\boldsymbol{\phi}^{\mathrm{base}} + \boldsymbol{\beta}$ raises~MAE to~$4.02$\,dBm, consistent with the analysis in Corollary~\ref{cor:affine}: additive conditioning lacks the multiplicative structure needed to represent~$\eta_{\mathrm{eff}}$.
\ding{185}~Removing the occlusion features~$T_k$ and~$\bar{\rho}_k$ raises~MAE to~$3.95$\,dBm, showing that the last-segment transmittance product contributes meaningful amplitude information beyond per-scatterer direction and distance.

The global-only branch provides an environment-map-like~RF variant.
Combining it with local conditioning reduces~BLE~MAE from~$4.11$ to~$3.57$\,dBm.
Appendix~\ref*{app_scaling_rx} shows that accuracy at unseen receivers improves with more training receivers, with diminishing returns beyond about~$12$ in the evaluated scene.
Appendix~\ref*{app_ref_sensitivity} gives a~$0.55$\,dBm spread across single-receiver Stage~I reference choices, and Appendices~\ref*{app_grid_resolution}--\ref*{app_capacity} report grid-resolution and conditioning-width sensitivity.

\vspace{\mylen}
\subsection{Application Study: Sparse-Deployment~BLE Localization}\label{sec_application}
\vspace{\mylen}

\begin{wrapfigure}{r}{0.4\linewidth}
\centering
\includegraphics[width=\linewidth]{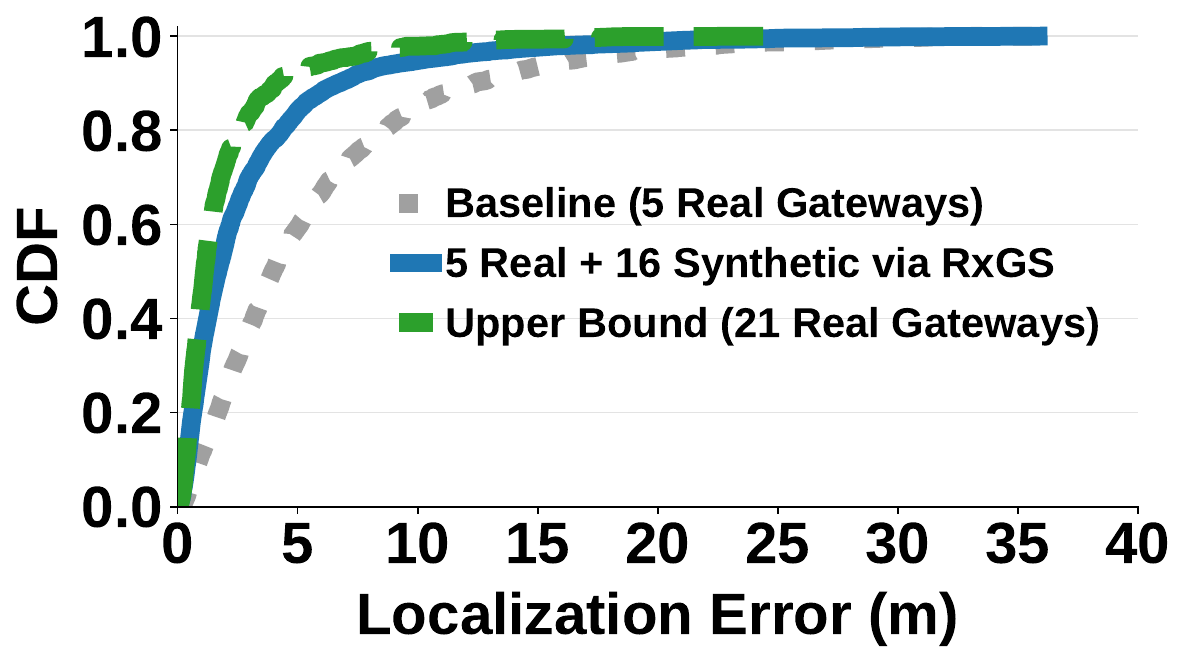}
\vspace{-0.2in}
\caption{Localization:~\ourSystem bridges $71\%$ of the sparse-to-dense gap.}
\label{fig_application}
\end{wrapfigure}

Fingerprint-based indoor localization trades off accuracy against gateway deployment cost.
Dense deployments provide high accuracy at high cost, whereas sparse deployments reduce cost at the expense of accuracy.
We show that~\ourSystem bridges this gap by synthesizing fingerprints at virtual gateway positions.
Following the protocol of~\S\ref{sec_generalization}, we assume only~$N\!=\!5$ gateways are physically installed and use~\ourSystem to synthesize~RSSI at the remaining~$16$ positions, forming a~$21$-dimensional fingerprint.
The~$16$ virtual positions coincide with measured gateway sites excluded from training, allowing direct ground-truth validation of their synthesized fingerprints.
Each of the~$1{,}200$ test transmitters is localized via weighted~$K$-nearest-neighbor~($K\!=\!5$) against a database of~$4{,}800$ training positions.

Figure~\ref{fig_application} shows the cumulative distribution function~(CDF) of localization error.
The sparse~$5$-gateway baseline yields a mean error of~$5.61 \pm 0.23$\,m.
Augmenting it with~$16$~\ourSystem-synthesized fingerprints reduces the mean error to~$2.98 \pm 0.67$\,m, a~$47\%$ improvement over the sparse baseline.
This closes~$71\%$ of the gap to the full~$21$-gateway oracle, which achieves~$1.91$\,m.

For~$5$,~$8$,~$11$,~$14$, and~$17$ physical gateways, mean errors over~$20$ random subsets are~$5.51$,~$4.11$,~$3.13$,~$2.69$, and~$2.27$\,m, respectively.
Table~\ref*{tab:localization_gateway_sweep} gives~$1.91$\,m for the full~$21$-gateway reference.
The fixed five-gateway layout in Figure~\ref*{fig_application} yields~$5.61$\,m, compared with~$5.51$\,m averaged across the random layouts in Table~\ref*{tab:localization_gateway_sweep}.
Appendix~\ref*{app_network_planning} studies network planning.

\begin{table}[H]
\centering
\makeatletter\def\fnum@table{Table~\thetable}\makeatother
\caption{BLE localization error versus number of physical gateways.
For~$5$--$17$ gateways, errors are averaged over~$20$ random subsets.
The~$21$-gateway result uses all measurements.}
\label{tab:localization_gateway_sweep}
\small
{
\begin{tabular}{lcccccc}
\toprule
Physical gateways & $5$ & $8$ & $11$ & $14$ & $17$ & $21$ \\
\midrule
Mean error (m)$\downarrow$ & $5.51$ & $4.11$ & $3.13$ & $2.69$ & $2.27$ & $1.91$ \\
\bottomrule
\end{tabular}
}
\end{table}
Table~\ref*{tab:localization_gateway_sweep} shows smaller error reductions as physical gateways are added, motivating spatially diverse placement and synthesized fingerprints within the surveyed region.

\vspace{\mylen}
\section{Conclusion and Future Work}\label{sec_conclusion}
\vspace{\mylen}

\ourSystem learns shared~3DGS geometry and receiver-conditioned radiance for fixed receivers in a surveyed scene.
It predicts at receiver positions excluded from training using one model instead of~$N$ receiver-specific checkpoints.
On~RSSI, it improves~RF-3DGS and~GSRF, while its~WRF-GS+ variant eliminates two high-error outliers with a~$0.33$\,dBm increase in mean~MAE.
On spatial spectrum and~CSI, it trails per-receiver baselines by~$0.8$ to~$1.7$\,dB.
Training, inference, and storage costs decrease by up to~$45\times$,~$7.6\times$, and~$24\times$, respectively.

The method assumes a static surveyed scene, requires a reference-receiver survey in each new environment, and models last-segment receiver effects within the~3DGS rendering abstraction.
Angular-spectrum experiments use simulated data, while measured~CSI spans eight effective receivers in one room.
Future work will evaluate measured multi-receiver spectra, larger and dynamic scenes, full-multipath variation, computational scaling, and cross-environment transfer.

\section*{Acknowledgments and Disclosures}
\label{sec:acknowledgement}

The research reported in this paper was sponsored in part by the National Science Foundation (award \# ECCS-2525614) and the U.S. Army DEVCOM Army Research Laboratory (award \# W911NF1720196).
Mani Srivastava was partially supported by the Mukund Padmanabhan Term Chair at UCLA.
Kang Yang was partially supported by the UCLA Institute for Digital Research and Education~(IDRE) fellowship.
The views and conclusions contained in this document are those of the authors and should not be interpreted as representing the official policies, either expressed or implied, of the funding agencies.
The authors declare no competing interests.

\medskip
{
\small
\bibliographystyle{unsrt}

}

\clearpage

\setcounter{tocdepth}{2}
\renewcommand{\contentsname}{Appendix Contents}
\makeatletter
\makeatother
\tableofcontents
\thispagestyle{empty}
\clearpage

\addtocontents{toc}{\protect\setcounter{tocdepth}{2}}

\appendix
\section{Proof of Geometry-Radiance Decomposition}\label{app_decomposition}

We prove that in the~3DGS-based~RF rendering framework, the receiver position~$\mathbf{r}$ influences the rendered signal only through the radiance coefficients, while all remaining attributes are receiver-independent.

\textbf{Notation.}
We use~$\mathbf{s}_k$ for scatterer positions in this appendix to distinguish the physical scatterer from the~Gaussian mean~$\mathbf{p}_k$.
Under Assumption~\ref{asm:per_scatterer},~$\mathbf{s}_k$ and~$\mathbf{p}_k$ coincide, so the two symbols are interchangeable.

\begin{definition}[Multipath Channel Model]\label{def:multipath}
The received signal at receiver~$\mathbf{r}$ from transmitter~$\mathbf{t}$ in a scene with~$K$ scatterers at fixed positions~$\left\{\mathbf{s}_k\right\}_{k=1}^{K}$ is:
\begin{equation}\label{eq:multipath}
    S\left(\mathbf{t}, \mathbf{r}\right) = \sum_{l=1}^{L} h_l\left(\mathbf{t}, \mathbf{r}\right),
\end{equation}
where~$L$ is the number of propagation paths, each traversing an ordered sequence of~$P_l \geq 0$ scatterers~$\left(\mathbf{s}_{l,1}, \ldots, \mathbf{s}_{l,P_l}\right)$.
A direct line-of-sight path corresponds to~$P_l = 0$ with a single segment~$d_{l,0} = \left\|\mathbf{t} - \mathbf{r}\right\|$.
The complex coefficient of the~$l$-th path is:
\begin{equation}\label{eq:path_coeff}
    h_l\left(\mathbf{t}, \mathbf{r}\right) = \underbrace{\prod_{i=1}^{P_l} \Gamma\left(\mathbf{s}_{l,i}\right)}_{\text{reflection coefficients}} \cdot \underbrace{\prod_{j=0}^{P_l} \frac{\lambda}{4\pi d_{l,j}}}_{\text{free-space path loss}} \cdot \underbrace{e^{-j\frac{2\pi}{\lambda}\sum_{j=0}^{P_l} d_{l,j}}}_{\text{phase accumulation}},
\end{equation}
with segment distances~$d_{l,0} = \left\|\mathbf{t} - \mathbf{s}_{l,1}\right\|$,~$d_{l,j}$ between consecutive scatterers, and~$d_{l,P_l} = \left\|\mathbf{s}_{l,P_l} - \mathbf{r}\right\|$ for~$P_l \geq 1$.
\end{definition}

\begin{definition}[RF Rendering Equation]\label{def:rendering}
Given~3D Gaussians~$\left\{\mathcal{G}_k\right\}_{k=1}^{K}$ with attributes~$\mathcal{G}_k = \left(\mathbf{p}_k, \mathbf{C}_k, \phi_k, \tau_k\right)$, the rendered signal for a ray in direction~$\boldsymbol{\omega}$ emitted from~$\mathbf{t}$ is:
\begin{equation}\label{eq:render}
    R\left(\boldsymbol{\omega}\right) = \sum_{k=1}^{K_{\boldsymbol{\omega}}} g_k\left(\boldsymbol{\omega};\, \mathbf{p}_k, \mathbf{C}_k\right) \cdot \phi_k\left(\hat{\mathbf{d}}_k\right) \cdot \prod_{m=1}^{k-1} \tau_m,
\end{equation}
where~$g_k$ is the Gaussian weight measuring ray-Gaussian overlap,~$\phi_k\left(\hat{\mathbf{d}}_k\right) \in \mathbb{C}$ is the directional radiance along~$\hat{\mathbf{d}}_k = \frac{\mathbf{p}_k - \mathbf{t}}{\left\|\mathbf{p}_k - \mathbf{t}\right\|}$, and~$\prod_{m=1}^{k-1} \tau_m$ is the cumulative transmittance along the ray.
The transmittance~$\tau_k$ models intrinsic material attenuation at~$\mathbf{s}_k$; directional path-length effects are carried by~$g_k$.
The total rendered signal~$\hat{S}$ aggregates~$R\left(\boldsymbol{\omega}\right)$ over a finite set of sampled ray directions, and represents all signal contributions at the receiver, including the direct line-of-sight component absorbed by the Gaussians along the~$\mathbf{t} \to \mathbf{r}$ line.
\end{definition}

\begin{definition}[Receiver-Independence]\label{def:rx_indep}
A scene attribute~$\xi$ is \emph{receiver-independent} if~$\frac{\partial \xi}{\partial \mathbf{r}} = \mathbf{0}$ for all~$\mathbf{r}$, and \emph{receiver-dependent} if~$\xi\left(\mathbf{r}_1\right) \neq \xi\left(\mathbf{r}_2\right)$ for some~$\mathbf{r}_1, \mathbf{r}_2$.
\end{definition}

\begin{assumption}[Static scene]\label{asm:scene}
Scatterer positions~$\left\{\mathbf{s}_k\right\}$, their spatial extents, and material properties~$\left\{\Gamma\left(\mathbf{s}_k\right)\right\}$ are fixed and do not depend on~$\mathbf{t}$ or~$\mathbf{r}$.
\end{assumption}

\begin{assumption}[Non-triviality]\label{asm:nontrivial}
There exist~$\mathbf{r}_1 \neq \mathbf{r}_2$ and~$\mathbf{t}$ such that~$S\left(\mathbf{t}, \mathbf{r}_1\right) \neq S\left(\mathbf{t}, \mathbf{r}_2\right)$.
\end{assumption}

\begin{assumption}[Per-scatterer modeling abstraction]\label{asm:per_scatterer}
Each Gaussian~$\mathcal{G}_k$ is treated as representing the aggregate signal contribution localized at a scene point~$\mathbf{s}_k$, with the Gaussian weight~$W_k$ taken to absorb the cumulative propagation and attenuation effects of all paths whose last interaction is at~$\mathbf{s}_k$, together with the direct line-of-sight contribution for Gaussians along the~$\mathbf{t} \to \mathbf{r}$ line.
This is a modeling abstraction adopted by~3DGS-based~RF methods to bridge the straight-ray rendering pipeline and the bouncing-path multipath model; under this abstraction, ~$W_k$ is identified with the multipath aggregate~$\Phi_k$ in Equation~\eqref{eq:partition}, as discussed in Remark~\ref{rem:scope}.
\end{assumption}

\begin{assumption}[Rendering-model representability]\label{asm:fit}
At the learned optimum,~$\hat{S}\left(\mathbf{t}, \mathbf{r}\right) = S\left(\mathbf{t}, \mathbf{r}\right)$ in the limit of perfect representability.
\end{assumption}

\begin{proposition}[Geometry-Radiance Decomposition]\label{prop:decomposition}
Under Assumptions~\ref{asm:scene}--\ref{asm:fit}:
\begin{enumerate}
    \item The geometric attributes~$\left\{\mathbf{p}_k, \mathbf{C}_k, \tau_k\right\}$ are receiver-independent.
    \item The radiance~$\phi_k = \phi_k\left(\hat{\mathbf{d}}_k, \mathbf{r}\right)$ is receiver-dependent, with receiver-gradient given by Equation~\eqref{eq:final_grad} under the per-scatterer identification.
    
    \item Under the per-scatterer identification, the radiance is the only term in Equation~\eqref{eq:render} that carries receiver-dependence; all remaining factors are receiver-independent.
    
\end{enumerate}
\end{proposition}

\begin{proof}

\textbf{Step~1 (Per-scatterer decomposition).}
Each propagation path terminates at the receiver, with its final segment originating from a scatterer~$\mathbf{s}_{l,P_l}$ for~$P_l \geq 1$, or directly from~$\mathbf{t}$ for the line-of-sight path~$P_l = 0$.
Let~$\mathcal{L}_k$ denote the set of paths whose last interaction is localized at~$\mathbf{s}_k$, including the LoS path assigned to the Gaussian nearest the~$\mathbf{t} \to \mathbf{r}$ line under Assumption~\ref{asm:per_scatterer}.
The sets~$\left\{\mathcal{L}_k\right\}$ partition all paths, yielding:
\begin{equation}\label{eq:partition}
    S\left(\mathbf{t}, \mathbf{r}\right) = \sum_{k=1}^{K} \Phi_k\left(\mathbf{t}, \mathbf{r}\right), \qquad \Phi_k \triangleq \sum_{l \in \mathcal{L}_k} h_l\left(\mathbf{t}, \mathbf{r}\right).
\end{equation}

For the rendering equation, define the per-ray weight~$w_k\left(\boldsymbol{\omega}\right) = g_k\left(\boldsymbol{\omega};\, \mathbf{p}_k, \mathbf{C}_k\right) \prod_{m=1}^{k-1} \tau_m$.
Aggregating over rays:
\begin{equation}\label{eq:render_reorder}
    \hat{S}\left(\mathbf{t}, \mathbf{r}\right) = \sum_{k=1}^{K} W_k \cdot \phi_k\left(\hat{\mathbf{d}}_k, \mathbf{r}\right), \qquad W_k \triangleq \sum_{\boldsymbol{\omega}} w_k\left(\boldsymbol{\omega}\right).
\end{equation}
Under Assumptions~\ref{asm:per_scatterer} and~\ref{asm:fit}, the per-scatterer identification is:
\begin{equation}\label{eq:term_match}
W_k \cdot \phi_k\!\left(\hat{\mathbf{d}}_k, \mathbf{r}\right) = \Phi_k\!\left(\mathbf{t}, \mathbf{r}\right), \quad \forall\, k.
\end{equation}
This identification is the formal bridge between the rendering-equation weight~$W_k$ and the multipath aggregate~$\Phi_k$, and is the only point at which the two models are joined.

\textbf{Step~2 (Geometric attributes are receiver-independent).}
By Assumption~\ref{asm:scene},~$\mathbf{p}_k$,~$\mathbf{C}_k$, and~$\mathbf{s}_k$ are fixed, and~$\boldsymbol{\omega}$ (defined on the emission surface centered at~$\mathbf{t}$) is independent of~$\mathbf{r}$.
Chain rule gives~$\partial g_k/\partial \mathbf{r} = \mathbf{0}$ and~$\partial \tau_k/\partial \mathbf{r} = 0$, so:
\begin{equation}\label{eq:Wk_indep}
    \frac{\partial W_k}{\partial \mathbf{r}} = \mathbf{0}, \qquad \frac{\partial \hat{\mathbf{d}}_k}{\partial \mathbf{r}} = \mathbf{0}.
\end{equation}
This establishes claim~1.

\textbf{Step~3 (Physical form of the radiance receiver-gradient).}
Under the per-scatterer identification, Equation~\eqref{eq:term_match} with~$W_k \neq 0$ gives~$\phi_k = \Phi_k / W_k$.
Since~$W_k$ is receiver-independent:
\begin{equation}\label{eq:dphi_dr}
    \frac{\partial \phi_k}{\partial \mathbf{r}} = \frac{1}{W_k} \sum_{l \in \mathcal{L}_k} \frac{\partial h_l}{\partial \mathbf{r}}.
\end{equation}
From Equation~\eqref{eq:path_coeff},~$h_l$ depends on~$\mathbf{r}$ only through the last-segment distance~$d_{l,P_l} = \left\|\mathbf{o}_l - \mathbf{r}\right\|$, where~$\mathbf{o}_l = \mathbf{s}_{l,P_l}$ is the last scatterer for~$P_l \geq 1$ and~$\mathbf{o}_l = \mathbf{t}$ for the line-of-sight path~$P_l = 0$.
The chain rule with~$\partial d_{l,P_l}/\partial \mathbf{r} = -(\mathbf{o}_l - \mathbf{r})/\left\|\mathbf{o}_l - \mathbf{r}\right\|$ and~$\partial h_l/\partial d_{l,P_l} = h_l \left(-1/d_{l,P_l} - j2\pi/\lambda\right)$ yields:
\begin{equation}\label{eq:final_grad}
    \frac{\partial \phi_k}{\partial \mathbf{r}} = \frac{1}{W_k} \sum_{l \in \mathcal{L}_k} h_l \cdot \left(\frac{1}{d_{l,P_l}} + \frac{j2\pi}{\lambda}\right) \cdot \frac{\mathbf{o}_l - \mathbf{r}}{\left\|\mathbf{o}_l - \mathbf{r}\right\|}.
\end{equation}
The~$1/d_{l,P_l}$ term is the logarithmic derivative of free-space path loss, and~$2\pi/\lambda$ is the phase accumulation rate, both standard for last-leg propagation under the~Friis law.

\textbf{Step~4 (Radiance carries all receiver-dependence).}
By Assumptions~\ref{asm:fit} and~\ref{asm:nontrivial}, there exist~$\mathbf{r}_1 \neq \mathbf{r}_2$ and~$\mathbf{t}$ with~$\hat{S}\left(\mathbf{t}, \mathbf{r}_1\right) \neq \hat{S}\left(\mathbf{t}, \mathbf{r}_2\right)$.
By Equation~\eqref{eq:render_reorder} and the receiver-independence of~$W_k$ and~$\hat{\mathbf{d}}_k$ from Step~2:
\begin{equation}\label{eq:S_diff}
    \sum_{k=1}^{K} W_k \cdot \left[\phi_k\left(\hat{\mathbf{d}}_k, \mathbf{r}_1\right) - \phi_k\left(\hat{\mathbf{d}}_k, \mathbf{r}_2\right)\right] \neq 0.
\end{equation}
Hence there exists~$k^\star$ with~$W_{k^\star} \neq 0$ and~$\phi_{k^\star}\left(\hat{\mathbf{d}}_{k^\star}, \mathbf{r}_1\right) \neq \phi_{k^\star}\left(\hat{\mathbf{d}}_{k^\star}, \mathbf{r}_2\right)$, establishing claim~2.
Since Step~2 shows all other factors in Equation~\eqref{eq:render} are receiver-independent, the radiance is the only term that can carry receiver-dependence, establishing claim~3.

Combining the three claims, the rendering equation factorizes as:
\begin{equation}\label{eq:final_decomp}
    \hat{S}\left(\mathbf{t}, \mathbf{r}\right) = \mathcal{R}\!\left(\mathbf{t},\; \underbrace{\left\{\mathbf{p}_k, \mathbf{C}_k, \tau_k\right\}}_{\partial/\partial\,\mathbf{r}\;=\;\mathbf{0}},\; \underbrace{\left\{\phi_k\left(\hat{\mathbf{d}}_k, \mathbf{r}\right)\right\}}_{\partial/\partial\,\mathbf{r}\;\neq\;\mathbf{0}}\right),
\end{equation}
which justifies freezing the geometric attributes after Stage~I and conditioning only the radiance on~$\mathbf{r}$ in Stage~II.
\end{proof}

\begin{remark}[Scope of Assumption~\ref{asm:per_scatterer}]\label{rem:scope}
The straight-ray rendering equation and the bouncing-path multipath model are not structurally identical: single-bounce paths align naturally with Gaussians along~$\mathbf{t} \to \mathbf{s}_k$, while multi-bounce and line-of-sight paths do not.
Assumption~\ref{asm:per_scatterer} adopts the phenomenological abstraction by which the~3DGS-RF framework bridges these models, assigning each path to the Gaussian whose localized region captures its final signal contribution.
This assumption is invoked only to derive the explicit gradient form in Step~3.
Claims~1, 2 (existence), and~3 rest only on Assumptions~\ref{asm:scene},~\ref{asm:nontrivial}, and~\ref{asm:fit}, and hold regardless of whether the per-scatterer abstraction is accepted.
We frame the result as a structural decomposition of the rendering equation rather than a physical theorem.
Its role in the paper is to motivate the two-stage architecture rather than to make claims about the underlying electromagnetic propagation.
\end{remark}

\section{Proof of Last-Segment Factorization}\label{app_factorization}
\begin{proof}
Equation~\eqref{eq:term_match} with~$W_k \neq 0$ gives~$\phi_k = \Phi_k / W_k$, and substituting~$\Phi_k = \sum_{l \in \mathcal{L}_k} h_l$:
\begin{equation}
    \phi_k(\hat{\mathbf{d}}_k, \mathbf{r}) = \frac{1}{W_k} \sum_{l \in \mathcal{L}_k} h_l(\mathbf{t}, \mathbf{r}).
\end{equation}
By the definition of~$\mathcal{L}_k$ (paths whose last interaction is at~$\mathbf{s}_k$), the last scatterer is~$\mathbf{s}_{l,P_l} = \mathbf{s}_k$ for all~$l \in \mathcal{L}_k$, so the last-segment distance~$d_{l,P_l} = \|\mathbf{s}_k - \mathbf{r}\|$ is shared across all paths in~$\mathcal{L}_k$.
Expanding~$h_l$ from Equation~\eqref{eq:path_coeff} and separating the last-segment terms:
\begin{equation}
    h_l(\mathbf{t}, \mathbf{r}) 
    = \underbrace{\prod_{i=1}^{P_l} \Gamma(\mathbf{s}_{l,i}) 
    \cdot \prod_{j=0}^{P_l-1} \frac{\lambda}{4\pi d_{l,j}} 
    \cdot e^{-j\frac{2\pi}{\lambda}\sum_{j=0}^{P_l-1} d_{l,j}}}_{h_l^{\mathrm{up}}(\mathbf{t})}
    \cdot \underbrace{\frac{\lambda}{4\pi\|\mathbf{s}_k - \mathbf{r}\|} 
    \cdot e^{-j\frac{2\pi}{\lambda}\|\mathbf{s}_k - \mathbf{r}\|}}_{\eta(\mathbf{s}_k, \mathbf{r})},
\end{equation}
where~$h_l^{\mathrm{up}}(\mathbf{t})$ contains all reflection coefficients, path losses, and phase accumulations from the transmitter to scatterer~$k$ along path~$l$.
Since~$\eta(\mathbf{s}_k, \mathbf{r})$ is identical for all~$l \in \mathcal{L}_k$, it factors out of the sum:
\begin{equation}
    \phi_k(\hat{\mathbf{d}}_k, \mathbf{r}) 
    = \eta(\mathbf{s}_k, \mathbf{r}) \cdot \underbrace{\frac{1}{W_k} \sum_{l \in \mathcal{L}_k} h_l^{\mathrm{up}}(\mathbf{t})}_{\phi_k^{\mathrm{base}}(\hat{\mathbf{d}}_k, \mathbf{t})}.
\end{equation}
This establishes the free-space factorization~$\phi_k = \eta \cdot \phi_k^{\mathrm{base}}$.

When the last segment from~$\mathbf{s}_k$ to~$\mathbf{r}$ traverses intervening scatterers, the cumulative transmittance~$\prod_{m \in \mathcal{I}(k, \mathbf{r})} \tau_m$ from the rendering equation (Equation~\eqref{eq:render}) attenuates the signal further.
Under Assumption~\ref{asm:per_scatterer}, this attenuation acts only on the last segment, so it multiplies~$\eta$ without affecting~$\phi_k^{\mathrm{base}}$, yielding the effective propagation factor:
\begin{equation}
    \eta_{\mathrm{eff}}(\mathbf{s}_k, \mathbf{r}) 
    = \eta(\mathbf{s}_k, \mathbf{r}) \cdot \prod_{m \in \mathcal{I}(k, \mathbf{r})} \tau_m.
\end{equation}
The full factorization follows:
\begin{equation}
    \phi_k(\hat{\mathbf{d}}_k, \mathbf{r}) = \eta_{\mathrm{eff}}(\mathbf{s}_k, \mathbf{r}) \cdot \phi_k^{\mathrm{base}}(\hat{\mathbf{d}}_k, \mathbf{t}).
\end{equation}
\end{proof}

\section{Proofs of Corollaries~\ref{cor:global_local}--\ref{cor:shared}}\label{app_corollary_proofs}

Notation follows Appendix~\ref{app_decomposition}: $\mathbf{s}_k$ denotes scatterer positions, interchangeable with the~Gaussian mean~$\mathbf{p}_k$ under Assumption~\ref{asm:per_scatterer}.

\subsection{Proof of Corollary~\ref{cor:global_local}}
\begin{proof}
Define the mean propagation effect and per-scatterer residual:
\begin{equation}
    \bar{\eta}(\mathbf{r}) = \frac{1}{K} \sum_{k=1}^{K} \eta_{\mathrm{eff}}(\mathbf{s}_k, \mathbf{r}), \qquad
    \delta_k(\mathbf{r}) = \eta_{\mathrm{eff}}(\mathbf{s}_k, \mathbf{r}) - \bar{\eta}(\mathbf{r}).
\end{equation}
By construction,~$\eta_{\mathrm{eff}} = \bar{\eta} + \delta_k$, and substituting into Proposition~\ref{prop:factorization} yields the decomposition in Equation~\eqref{eq:anova}.

It remains to show~$\delta_k \neq 0$ for almost all~$\mathbf{r}$.
By Equation~\eqref{eq:eta_eff},~$\eta_{\mathrm{eff}}(\mathbf{s}_{k_1}, \mathbf{r}) = \eta_{\mathrm{eff}}(\mathbf{s}_{k_2}, \mathbf{r})$ requires (in the best case where occlusion terms match)~$\|\mathbf{s}_{k_1} - \mathbf{r}\| = \|\mathbf{s}_{k_2} - \mathbf{r}\|$, whose solution set is the perpendicular bisector hyperplane of~$\mathbf{s}_{k_1}$ and~$\mathbf{s}_{k_2}$ --- a set of Lebesgue measure zero in~$\mathbb{R}^3$.
Therefore,~$\eta_{\mathrm{eff}}(\mathbf{s}_{k_1}, \mathbf{r}) \neq \eta_{\mathrm{eff}}(\mathbf{s}_{k_2}, \mathbf{r})$ for almost all~$\mathbf{r}$, so~$\delta_k \neq 0$ in general.
The global term~$\bar{\eta}(\mathbf{r}) \cdot \phi_k^{\mathrm{base}}$ alone cannot capture per-scatterer variation, and the local term~$\delta_k(\mathbf{r}) \cdot \phi_k^{\mathrm{base}}$ alone cannot capture the shared component, so both branches are required.
\end{proof}

\subsection{Proof of Corollary~\ref{cor:affine}}

\begin{proof}

From Proposition~\ref{prop:factorization}, ~$\phi_k\!\left(\hat{\mathbf{d}}_k, \mathbf{r}\right) = \eta_{\mathrm{eff}}\!\left(\mathbf{s}_k, \mathbf{r}\right) \cdot \boldsymbol{\phi}^{\mathrm{base}}_k\!\left(\hat{\mathbf{d}}_k, \mathbf{t}\right)$, where~$\eta_{\mathrm{eff}} \in \mathbb{C}$.
We show that affine modulation suffices to represent this complex-valued multiplication, and that two simpler alternatives are insufficient in general.

\textbf{Step~1: Complex multiplication as a structured linear map.}
Write~$\eta_{\mathrm{eff}} = a + jb$ and let~$\phi_k^{\mathrm{base},(n)} = x_n + jy_n$ denote the~$n$-th component of the complex-valued radiance vector.
Component-wise, the product~$\eta_{\mathrm{eff}} \cdot \phi_k^{\mathrm{base},(n)}$ expands as:
\begin{equation}\label{eq:complex_mat}
    \begin{bmatrix} \mathrm{Re}(\phi_k^{(n)}) \\ \mathrm{Im}(\phi_k^{(n)}) \end{bmatrix}
    = \underbrace{\begin{bmatrix} a & -b \\ b & a \end{bmatrix}}_{\mathbf{M}(\eta_{\mathrm{eff}})}
    \begin{bmatrix} x_n \\ y_n \end{bmatrix}.
\end{equation}
The matrix~$\mathbf{M}$ simultaneously scales amplitude by~$|\eta_{\mathrm{eff}}| = \sqrt{a^2 + b^2}$ and rotates phase by~$\angle\eta_{\mathrm{eff}} = \mathrm{atan2}(b, a)$.
Any conditioning mechanism must be able to represent this joint amplitude-and-phase operation.

\textbf{Step~2: Receiver-shared additive conditioning is insufficient.}
Consider an additive model~$\phi_k = \phi_k^{\mathrm{base}} + \boldsymbol{\beta}(\mathbf{r})$, where~$\boldsymbol{\beta}$ depends on~$\mathbf{r}$ only (shared across all Gaussians, with no per-~$k$ dependence).
For the true radiance~$\phi_k = \eta_{\mathrm{eff}} \cdot \phi_k^{\mathrm{base}}$, this requires:
\begin{equation}
    \boldsymbol{\beta}(\mathbf{r}) = (\eta_{\mathrm{eff}}(\mathbf{s}_k, \mathbf{r}) - 1) \cdot \phi_k^{\mathrm{base}}.
\end{equation}
Since~$\boldsymbol{\beta}$ does not depend on~$k$, this requires~$(\eta_{\mathrm{eff}}(\mathbf{s}_{k_1}, \mathbf{r}) - 1) \cdot \phi_{k_1}^{\mathrm{base}} = (\eta_{\mathrm{eff}}(\mathbf{s}_{k_2}, \mathbf{r}) - 1) \cdot \phi_{k_2}^{\mathrm{base}}$ for all~$k_1, k_2$, which fails in general since both~$\eta_{\mathrm{eff}}$ and~$\phi_k^{\mathrm{base}}$ vary across scatterers.
Therefore, receiver-shared additive conditioning cannot represent the multiplicative structure.

\textbf{Step~3: Real-valued scaling is insufficient.}
Consider a real-valued scaling model~$\phi_k = \alpha \cdot \phi_k^{\mathrm{base}}$ with~$\alpha \in \mathbb{R}$.
Component-wise:
\begin{equation}
    \begin{bmatrix} \mathrm{Re}(\phi_k^{(n)}) \\ \mathrm{Im}(\phi_k^{(n)}) \end{bmatrix}
    = \begin{bmatrix} \alpha & 0 \\ 0 & \alpha \end{bmatrix}
    \begin{bmatrix} x_n \\ y_n \end{bmatrix}.
\end{equation}
This is a diagonal matrix that can only scale amplitude, not rotate phase.
Matching~$\mathbf{M}(\eta_{\mathrm{eff}})$ from Equation~\eqref{eq:complex_mat} requires~$b = 0$,~\ie~$\mathrm{Im}(\eta_{\mathrm{eff}}) = 0$.
From Equation~\eqref{eq:eta_eff},~$\mathrm{Im}(\eta_{\mathrm{eff}}) = 0$ only when~$\frac{2\pi}{\lambda}\|\mathbf{s}_k - \mathbf{r}\| = n\pi$ for integer~$n$, which holds only at discrete distances.
Therefore, real-valued scaling fails for general receiver positions.

\textbf{Step~4: Complex affine modulation suffices.}
An element-wise affine modulation~$(1 + \boldsymbol{\alpha}) \odot \phi_k^{\mathrm{base}} + \boldsymbol{\beta}$ with~$\boldsymbol{\alpha}, \boldsymbol{\beta} \in \mathbb{C}^d$ parameterizes complex scale and shift for each component of~$\phi_k^{\mathrm{base}}$.
Setting~$\boldsymbol{\alpha} = (\eta_{\mathrm{eff}} - 1) \cdot \mathbf{1}$ and~$\boldsymbol{\beta} = \mathbf{0}$ recovers the exact product:
\begin{equation}
    (1 + \boldsymbol{\alpha}) \odot \phi_k^{\mathrm{base}} = \eta_{\mathrm{eff}} \cdot \phi_k^{\mathrm{base}}.
\end{equation}
The additive term~$\boldsymbol{\beta}$ provides additional capacity to absorb residual effects not captured by~$\eta_{\mathrm{eff}}$ alone, such as multi-bounce interactions on the last segment.
Since affine modulation succeeds where both receiver-shared additive conditioning~(Step~2) and real-valued scaling~(Step~3) fail, it preserves the multiplicative structure of the physics while providing slack for residual effects.
\end{proof}

\subsection{Proof of Corollary~\ref{cor:shared}}
\begin{proof}
We show that the scatterer index~$k$ carries no information beyond the per-scatterer inputs~$(d_k, T_k)$ required to evaluate~$\eta_{\mathrm{eff}}$, and therefore separate per-scatterer networks cannot improve over a single shared network.

\textbf{Step~1: Index-independence.}
From Equation~\eqref{eq:eta_eff}, the effective propagation factor for any scatterer~$k$ is:
\begin{equation}
    \eta_{\mathrm{eff}}(\mathbf{s}_k, \mathbf{r}) = f(d_k, T_k) = \frac{\lambda}{4\pi d_k} \cdot e^{-j\frac{2\pi}{\lambda} d_k} \cdot T_k,
\end{equation}
where~$d_k = \|\mathbf{s}_k - \mathbf{r}\|$ and~$T_k = \prod_{m \in \mathcal{I}(k, \mathbf{r})} \tau_m$.
The function~$f$ does not depend on the scatterer index~$k$: for any two scatterers~$k_1 \neq k_2$, if~$(d_{k_1}, T_{k_1}) = (d_{k_2}, T_{k_2})$, then~$\eta_{\mathrm{eff}}(\mathbf{s}_{k_1}, \mathbf{r}) = \eta_{\mathrm{eff}}(\mathbf{s}_{k_2}, \mathbf{r})$.
In the architectural implementation we augment the network inputs with the direction~$\hat{\mathbf{v}}_k$ and mean density~$\bar{\rho}_k$ to facilitate learning~$T_k$ from geometric cues, but these are not required by the theoretical form of~$f$.

\textbf{Step~2: Per-scatterer networks are redundant.}
Suppose we use separate networks~$\{\hat{f}_k\}_{k=1}^{K}$, one per scatterer, each approximating~$f$.
Since~$f$ is index-independent from Step~1, the optimal solution for every~$\hat{f}_k$ is identical:
\begin{equation}
    \hat{f}_{k_1}^* = \hat{f}_{k_2}^* = \cdots = \hat{f}_K^* = f, \quad \forall\, k_1, k_2.
\end{equation}
Therefore, at optimality, all per-scatterer networks converge to the same function, making the separate parameterizations redundant.

\textbf{Step~3: Shared network suffices.}
Define a single shared network~$\hat{f}$ that takes per-scatterer inputs for any~$k$.
From Step~1,~$\hat{f}$ must approximate~$f(d_k, T_k)$ over the union of all per-scatterer inputs:
\begin{equation}
    \hat{f} \approx f \quad \text{on} \quad \bigcup_{k=1}^{K} \{(d_k, T_k) : \mathbf{r} \in \mathbb{R}^3\}.
\end{equation}
Since~$f$ is continuous for~$d_k > 0$ and the input domain is bounded in any physical scene, the universal approximation theorem guarantees that a single network~$\hat{f}$ with sufficient capacity can approximate~$f$ to arbitrary precision.
Therefore, a single shared network with per-scatterer inputs suffices for all~$k$.
\end{proof}

\section{Single-Kernel Multi-Receiver Rasterization}
\label{app_multirx_impl}

The core mechanism behind~\ourSystem's multi-receiver inference is a single~CUDA call that renders all~$N$ receivers at once.
This exploits the fact that the sort and per-Gaussian projection depend only on the transmitter, not on the receivers.
For one transmitter at~$\mathbf{t}$ serving~$N$ receivers~$\left\{\mathbf{r}_j\right\}_{j=1}^{N}$, the pipeline runs three stages, summarized in Algorithm~\ref{alg:multirx_render}.

\textbf{Transmitter-Side Preprocessing.}
We compute Gaussian projections, per-tile bounds, and the depth-sorted Gaussian list using only~$\mathbf{t}$ and the Gaussian state.
These results are receiver-independent and shared across all~$N$ rasterizations.
This stage therefore runs once per transmitter regardless of receiver count.

\textbf{Per-Receiver Directional Signal.}
We evaluate the~FLE/SH directional basis on the transmitter-side directions once, then reduce against the per-receiver coefficient tensor of shape~$\left(N,\, K,\, 2,\, L\right)$ in a single broadcasted operation.
This yields the per-receiver signal tensor~$\left\{s_k^{\left(j\right)}\right\} \in \mathbb{R}^{N \times K \times 2}$ in one batched call rather than~$N$ separate evaluations.

\textbf{Multi-Receiver Rasterization.}
A single~CUDA kernel is launched with grid dimension~$\text{gridDim} = \left(\text{tiles}_\theta,\, \text{tiles}_\varphi,\, N\right)$, placing the receiver index on the~$z$-axis.
Each thread block reads its receiver-specific signal slice and writes its receiver-specific output slot, while the shared transmitter-side state of projections, sorted list, and tile ranges is read by every block.
The backward kernel runs with the same~$\text{gridDim}.z = N$ to compute per-receiver~$\partial \mathcal{L} / \partial \phi_k^{\left(j\right)}$ and accumulate Gaussian-level gradients across all~$\left(j,\, \text{direction sample}\right)$ threads via~\texttt{atomicAdd}.

\begin{algorithm}[t]
\small
\caption{Multi-Receiver Spectrum Rendering for One Transmitter.}
\label{alg:multirx_render}
\begin{algorithmic}[1]
\Require Gaussian state~$\left\{\mathbf{p}_k,\, \mathbf{C}_k,\, \tau_k,\, \phi_k^{\mathrm{base}}\right\}_{k=1}^{K}$
\Require Transmitter position~$\mathbf{t}$, receiver positions~$\left\{\mathbf{r}_j\right\}_{j=1}^{N}$
\Require Angular grid of size~$H_\theta \times W_\varphi$ over elevation and azimuth
\Ensure Spectra~$\left\{\hat{S}^{\left(j\right)} \in \mathbb{R}^{2 \times H_\theta \times W_\varphi}\right\}_{j=1}^{N}$ (real and imaginary)
\State // Stage 1: Transmitter side, computed once
\State Project Gaussians onto the unit sphere centered at~$\mathbf{t}$ and compute per-Gaussian direction~$\hat{\mathbf{d}}_k$ and angular tile bounds \label{ln:proj}
\State Build per-tile depth-sorted Gaussian list via one~\texttt{cub::SortPairs} call over~$\left(\text{tile},\, \text{depth}\right)$ keys \label{ln:sort}
\State // Stage 2: Per-receiver coefficients, one batched call
\State $\left\{\phi_k^{\left(j\right)}\right\} \gets \texttt{rxcond.forward\_batched}\!\left(\phi_k^{\mathrm{base}},\, \left\{\mathbf{r}_j\right\},\, \left\{\mathbf{p}_k\right\}\right)$ \Comment{shape~$\left(N,\, K,\, 2,\, L\right)$} \label{ln:rxcond}
\State $\left\{s_k^{\left(j\right)}\right\} \gets \texttt{eval\_fle}\!\left(B\!\left(\hat{\mathbf{d}}_k\right),\, \left\{\phi_k^{\left(j\right)}\right\}\right)$ \Comment{basis~$B$ shared, reduce over~$L$} \label{ln:fle}
\State // Stage 3: Single CUDA kernel with~$\text{gridDim}.z = N$
\State \textbf{launch}~\texttt{render\_multirx<<<($\text{tiles}_\theta$, $\text{tiles}_\varphi$, $N$), \text{block}>>>}
\ForAll{angular tile~$\left(t_\theta,\, t_\varphi\right)$ and receiver~$j \in \left[1..N\right]$ in parallel}
  \ForAll{direction sample~$\left(\theta,\, \varphi\right)$ in tile~$\left(t_\theta,\, t_\varphi\right)$ in parallel}
    \State $\mathcal{T} \gets 1$, $C \gets 0$
    \ForAll{Gaussian~$k$ in tile's sorted list (front to back)}
      \State $w_k \gets \tau_k \cdot \exp\!\left(-\tfrac{1}{2} \boldsymbol{\Delta}_{k,\theta\varphi}^\top \mathbf{C}_k^{-1} \boldsymbol{\Delta}_{k,\theta\varphi}\right)$
      \State $C \gets C + \mathcal{T}\,w_k\, s_k^{\left(j\right)}$ \Comment{per-receiver signal}
      \State $\mathcal{T} \gets \mathcal{T}\,\left(1 - w_k\right)$
      \If{$\mathcal{T} < \epsilon$} \textbf{break} \EndIf
    \EndFor
    \State $\hat{S}^{\left(j\right)}_{\theta\varphi} \gets C$
  \EndFor
\EndFor
\end{algorithmic}
\end{algorithm}

\textbf{Shared Computation.}
The standard~3DGS rasterizer computes lines~\ref{ln:proj} and~\ref{ln:sort} once, then runs the per-direction alpha-blend once for a single spectrum.
For multi-receiver rendering, lines~\ref{ln:proj} and~\ref{ln:sort} depend only on the transmitter-anchored sphere rays and are reused across receivers.
The directional basis evaluated at~$\hat{\mathbf{d}}_k$ also depends only on the transmitter and is computed once on line~\ref{ln:fle}.
The per-receiver computation reduces to the directional signal~$s_k^{\left(j\right)}$, the~FLE/SH coefficient sequence~$\phi_k^{\left(j\right)}$ projected against the shared basis, producing one~$\left(N,\, K,\, 2\right)$ signal tensor.
The per-direction rendering then runs along an extra batch axis~$j \in \left[1..N\right]$, mapped to~$\text{blockIdx}.z$ at kernel launch.
Each receiver writes to its own output slot without atomic conflicts.

\section{Implementation Details}
\label{app_implementation}
We implement \ourSystem{} in PyTorch, with custom CUDA kernels for~RF signal tracing and rasterization to support the differentiable~RF rendering pipeline.
All experiments are conducted on a single~NVIDIA~H100~GPU.
This section describes the implementation of the~FLE-based scene representation.
Other representation variants~(SH) and signal modalities~(RSSI,~CSI) follow the same training pipeline and optimization strategy.
Dataset-specific hyperparameters are reported in Table~\ref{tab:hyperparams}.

\subsection{Scene Representation}
The~RF environment is represented as a set of~$K$~3D~Gaussians.
Each~Gaussian stores five attributes: position~$\mathbf{p} \in \mathbb{R}^3$, anisotropic scale~$\mathbf{s} \in \mathbb{R}^3$ (log-space), rotation quaternion~$\mathbf{q} \in \mathbb{R}^4$, sigmoid-activated transmittance~$\tau \in \left[0,1\right]$, and~FLE coefficients~$\mathbf{f} \in \mathbb{R}^{L \times C}$.
Here~$L = \left(\ell_{\max}+1\right)^2$ is the number of~FLE basis functions (matching the radiance-component count from the main text), and~$C$ is the number of signal channels.
The storage~$\mathbf{f}$ realizes the base coefficients~$\phi_k^{\mathrm{base}}$ from~\S\ref{sec_conditioning} in the implementation.
The covariance matrix is constructed as~$\boldsymbol{\Sigma} = \mathbf{R}\mathbf{S}\mathbf{S}^\top\mathbf{R}^\top$, where~$\mathbf{R}$ and~$\mathbf{S}$ are the rotation and scaling matrices built from~$\mathbf{q}$ and~$\mathbf{s}$.

\textbf{Initialization.}
We initialize up to~$50{,}000$~Gaussians from a point cloud sampled from the scene.
Scales are set to the log of nearest-neighbor distances, computed via a~CUDA~KNN module.
Rotations use identity quaternions, transmittance starts at~$\sigma^{-1}\left(0.1\right)$, and~FLE coefficients at zero.

\begin{table}[H]
\centering
\renewcommand{\arraystretch}{1.2}
\caption{Hyperparameter settings for \ourSystem{} across all datasets.}
\label{tab:hyperparams}
\begin{tabularx}{\textwidth}{>{\raggedright\arraybackslash}p{0.3\textwidth}
                              >{\raggedright\arraybackslash}p{0.2\textwidth}
                              >{\raggedright\arraybackslash}p{0.2\textwidth}
                              >{\raggedright\arraybackslash}X}
\toprule
\textbf{Parameter} & \textbf{BLE RSSI} & \textbf{RFID Spectrum} & \textbf{WiFi CSI} \\
\midrule
Initial Gaussians~($K_{\text{init}}$)        & $50{,}000$              & $50{,}000$              & $30{,}000$ \\
Max~FLE degree~($\ell_{\max}$)               & 9                       & 9                       & 4 \\
FLE ramp interval~($T_{\text{ramp}}$)        & 500                     & 500                     & 100 \\
Ray directions~($n_\varphi \times n_\theta$) & $36\times9$             & $360\times90$           & $72\times18$ \\
Sphere radius~($r_{\text{rx}}$)              & 4.22                    & 1.00                    & 1.00 \\
Position~LR (initial)                        & $5.05\times10^{-5}$     & $1.60\times10^{-4}$     & $2.78\times10^{-4}$ \\
Position~LR (final)                          & $2.06\times10^{-6}$     & $1.60\times10^{-6}$     & $1.60\times10^{-6}$ \\
Feature~LR                                   & $1.14\times10^{-3}$     & $5.00\times10^{-3}$     & $1.53\times10^{-2}$ \\
Rest~LR ratio                                & 0.16                    & 0.20                    & 0.80 \\
Transmittance~LR                             & $3.79\times10^{-2}$     & $1.00\times10^{-2}$     & $1.27\times10^{-3}$ \\
Scaling~LR                                   & $2.90\times10^{-3}$     & $5.00\times10^{-3}$     & $5.00\times10^{-3}$ \\
Rotation~LR                                  & $5.02\times10^{-4}$     & $1.00\times10^{-3}$     & $1.00\times10^{-3}$ \\
Densification threshold~($\tau_g$)           & $1.97\times10^{-4}$     & $2.00\times10^{-4}$     & $1.88\times10^{-4}$ \\
Transmittance reset interval                 & 4{,}000                 & 3{,}000                 & 2{,}000 \\
Fourier bands~($F$)                          & 6                       & 6                       & 5 \\
MLP hidden dim~($d$)                         & 64                      & 64                      & 256 \\
Component embedding dim~($d_c$)              & 16                      & 16                      & 16 \\
Occupancy probe samples~($S$)                & 16                      & 16                      & 16 \\
Grid resolution                              & $128^3$                 & $128^3$                 & $128^3$ \\
Stage~I iterations                           & 30K                     & 30K                     & 30K \\
Stage~II iterations                          & 100K                    & 60K                     & 100K \\
\bottomrule
\end{tabularx}
\end{table}

\subsection{Signal Rendering Pipeline}
For each~Gaussian, the renderer computes a complex signal pair from its directional~FLE coefficients along the~transmitter-to-Gaussian direction.
For a direction~$\hat{\mathbf{d}} = \left(\theta, \varphi\right)$ from the~transmitter to the~Gaussian center, the real and imaginary components of the radiated signal are:
\begin{align}
R_e &= \sum_{l=0}^{\ell_{\max}} \sum_{m=-l}^{l}
        \left[ a_{lm} \cos\left(m\varphi\right) - b_{lm} \sin\left(m\varphi\right) \right]
        N_{lm} P_l^{|m|}\left(\cos\theta\right) \\
R_i &= \sum_{l=0}^{\ell_{\max}} \sum_{m=-l}^{l}
        \left[ a_{lm} \sin\left(m\varphi\right) + b_{lm} \cos\left(m\varphi\right) \right]
        N_{lm} P_l^{|m|}\left(\cos\theta\right)
\end{align}
where~$a_{lm}, b_{lm} \in \mathbf{f}$ are learnable coefficients,~$P_l^{|m|}$ are associated~Legendre polynomials computed via batched three-term recurrence, and~$N_{lm} = \sqrt{\frac{\left(2l+1\right)}{4\pi} \cdot \frac{\left(l-|m|\right)!}{\left(l+|m|\right)!}}$ is the normalization factor.
The custom~CUDA rasterizer detailed in Appendix~\ref*{app_multirx_impl} combines these signals with each~Gaussian's covariance and transmittance through covariance-weighted accumulation.
The final output amplitude~$|R| = \sqrt{R_e^2 + R_i^2 + \epsilon}$ uses a numerical stabilizer~$\epsilon = 10^{-8}$.
The rendering grid uses a sphere of radius~$r_{\text{rx}}$ centered at the transmitter.

\begin{samepage}
\subsection{Loss Function}
Both training stages use a modality-specific loss:~$L_1$ for scalar~RSSI,~$L_2$ for complex-valued~CSI, and a composite loss for spatial spectrum.
For a rendered spectrum~$\hat{\mathbf{y}}$ and ground-truth spectrum~$\mathbf{y}$, the composite loss is:
\begin{equation}
\mathcal{L} = \left(1 - \lambda_{\text{ssim}} - \lambda_{\text{fft}}\right) \lVert \hat{\mathbf{y}} - \mathbf{y} \rVert_1
            + \lambda_{\text{ssim}} \left(1 - \text{SSIM}\left(\hat{\mathbf{y}}, \mathbf{y}\right)\right)
            + \lambda_{\text{fft}} \lVert \mathcal{F}\left(\hat{\mathbf{y}}\right) - \mathcal{F}\left(\mathbf{y}\right) \rVert_2^2
\end{equation}
Here~$\lambda_{\text{ssim}}$ and~$\lambda_{\text{fft}}$ weight the structural-similarity and Fourier-domain terms, respectively.
The operator~$\mathcal{F}$ denotes the~2D orthonormal~FFT, and~SSIM uses an~$11 \times 11$~Gaussian window.
\end{samepage}

\subsection{Stage~I: Geometry Learning}
Stage~I trains all~Gaussian parameters using measurements from a single reference receiver.
At each iteration, a random training sample from the reference receiver is selected, and the rendered signal is compared with the measurement using the modality-specific loss.

\textbf{Optimizer.}
We use~Adam with per-parameter learning rates listed in Table~\ref{tab:hyperparams}.
The position~LR decays exponentially from its initial to final value with a delay multiplier of~$0.01$.
Higher-degree~FLE coefficients are additionally scaled by~$0.2$ to prioritize low-frequency geometry early in training.
The~FLE degree is progressively increased from~$0$, incrementing every~$T_{\text{ramp}}$ iterations up to~$\ell_{\max}$.
Stage~I runs for~$30$K iterations total.

\textbf{Adaptive Densification.}
Gaussians are densified based on accumulated positional gradients every~$100$ iterations, from iteration~$500$ to the stage midpoint~(iteration~$15$K).
Gaussians exceeding the densification threshold~$\tau_g$ listed in Table~\ref*{tab:hyperparams} are cloned if their scale is below the scene extent threshold, or split into two children with scales reduced by factor~$0.8$ if above.
Gaussians with projected screen-space radius above a pixel threshold, or world-space extent exceeding~$0.1\times$ the scene extent, are pruned.
To avoid saturation, transmittance resets to~$\sigma^{-1}\left(0.01\right)$ as specified in Table~\ref*{tab:hyperparams}.

\subsection{Stage~II: Receiver-Conditioned Radiance}
With geometry frozen,~Stage~II trains the~FLE coefficients and the conditioning modules across all receivers.
At each iteration, a random training sample from an available receiver is selected.
The conditioning modules then modulate the base~FLE coefficients using its position, and the modality-specific loss compares the rendered signal with the corresponding measurement.

\textbf{Global Branch.}
Table~\ref*{tab:hyperparams} specifies~$F$ frequency bands per spatial dimension for the learnable~Fourier encoding of receiver position.
The bands start at~$2^k$ spacing and are optimized during training.
For each~FLE component, a~3-layer~ReLU~MLP with hidden dimension~$d$ takes the concatenation of the encoded~RX position, per-component degree and order features, and a learnable per-component embedding of dimension~$d_c$ as input.
The~MLP produces per-component scale and shift parameters that modulate the base~FLE coefficients.
This modulation is computed once per component and broadcast to all~$K$~Gaussians, giving~$\mathcal{O}\left(L\right)$ cost independent of~Gaussian count.

\textbf{Local Branch.}
At the start of~Stage~II, a~3D occupancy grid of resolution~$128^3$ is built by splatting each frozen~Gaussian into a voxel density field using its position, transmittance, and scale with a~$2\sigma$ radius.
For each~Gaussian~$k$, we cast a ray toward the receiver and probe the grid at~$S$ uniformly spaced points~(from~$t=0.05$ to~$t=0.95$) via nearest-neighbor lookup, computing path transmittance~$T_k = \prod_{s=1}^{S}\left(1 - \rho_s\right)$ and mean density~$\bar{\rho}_k$.
A shared~MLP takes~$\left[\hat{\mathbf{v}}_k,\; d_k,\; T_k,\; \bar{\rho}_k\right] \in \mathbb{R}^6$ as input and produces per-Gaussian scale and shift parameters.
Both~MLP final layers are zero-initialized, making the initial conditioning transform the identity.

\textbf{Optimizer.}
FLE coefficients and conditioning modules are trained with separate~Adam optimizers, with learning rates listed in Table~\ref{tab:hyperparams}.
Higher-degree~FLE coefficients are scaled by~$0.2$ as in~Stage~I.

\section{Dataset Details}
\label{app_dataset}

All three datasets use random seed~$8371$ for reproducible train/test splits.
TX denotes transmitter and~RX denotes receiver throughout this section.

\subsection{BLE RSSI}
The~BLE~RSSI dataset is collected in a real-world nursing home occupying~${\sim}15{,}000$\,ft$^2$~\cite{zhao2023nerf2}.
$N_{\mathrm{RX}} = 21$~BLE gateways operating at~$2.4$\,GHz are deployed throughout the facility, each measuring the received signal strength~(RSSI in~dBm) from a moving~BLE node.
TX positions are collected by random walks using~$30$~BLE nodes, yielding~$N_{\mathrm{TX}} = 6{,}000$~TX positions in total.
Each sample is a~$21$-dimensional tuple of~RSSI readings from the~$21$ receivers.
Readings below~$-100$\,dBm are excluded.
A random~80/20 split yields~$4{,}800$ training and~$1{,}200$ test positions.

\begin{figure}[H]
\centering
\includegraphics[width=\textwidth]{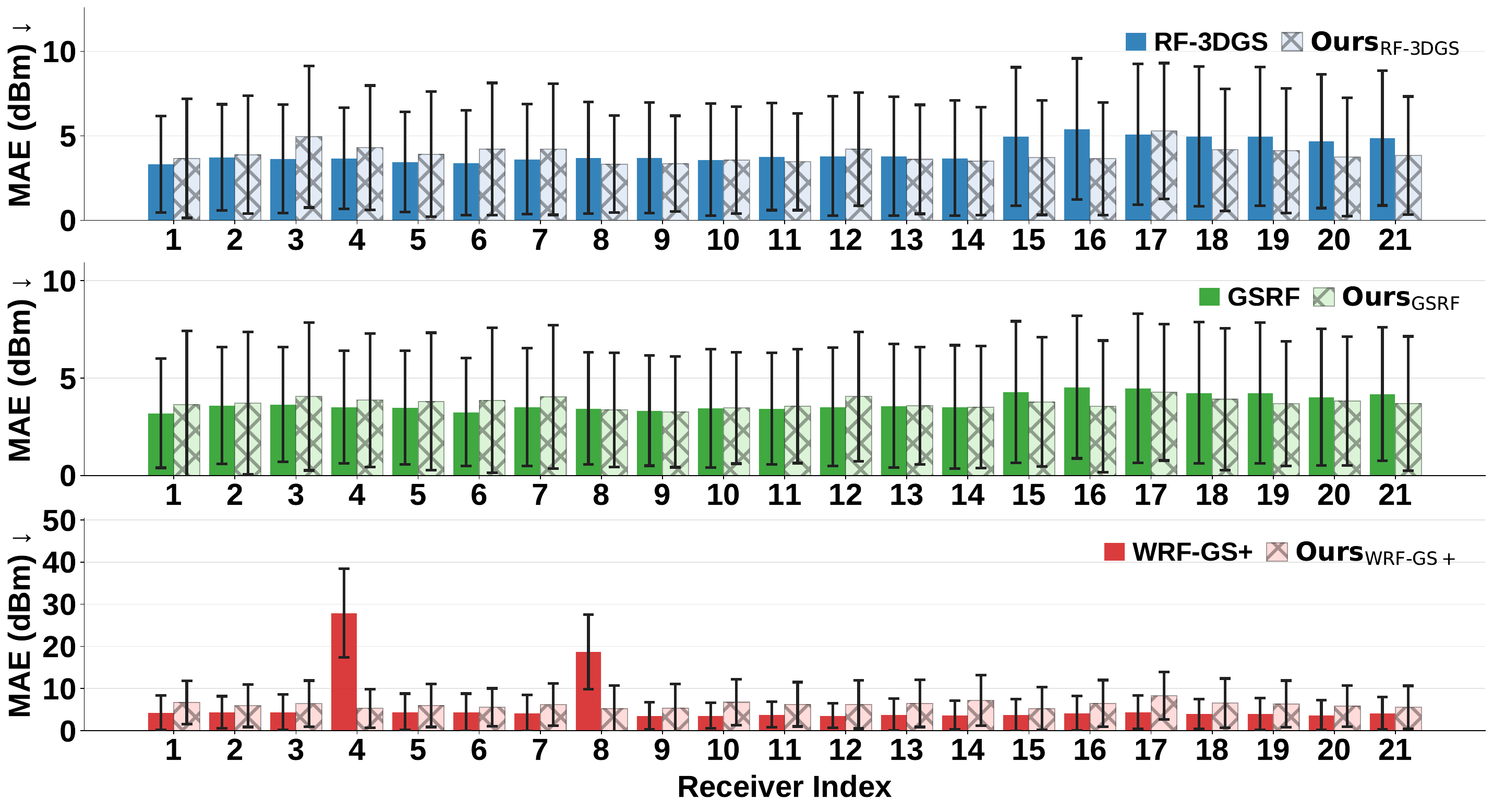}
\caption{Per-receiver~MAE breakdown on~BLE~RSSI for each baseline and its receiver-conditioned counterpart~(lower is better).
Adding our receiver conditioning produces more uniform performance across receivers and eliminates the catastrophic outliers seen in~WRF-GS+ at receivers~$4$ and~$8$.}
\label{fig:per_receiver_ble}
\end{figure}

\subsection{RFID Spatial Spectrum}
Collecting real-world spatial spectrum data requires a calibrated antenna array and precise~TX positioning across hundreds of positions, making large-scale data collection prohibitively expensive.
Furthermore, no public multi-receiver spatial spectrum dataset currently exists.
We therefore generate this dataset via simulation using the~NVIDIA~Sionna ray tracer~\cite{hoydis2022sionna, hoydis2023sionna}.
Sionna is a physics-based, differentiable ray tracing engine that models line-of-sight, specular reflection, diffuse scattering, and refraction, with material properties calibrated to~ITU-R standards.
Its~GPU acceleration enables large-scale dataset generation at simulation speeds.

The scene is a cluttered conference room~(8\,m $\times$ 6\,m $\times$ 3\,m), as shown in Figure~\ref{fig_layout}.
Each measurement is a~$90 \times 360$ grayscale image capturing the angular power spectrum over elevation and azimuth at~$2.4$\,GHz.
The dataset contains~$N_{\mathrm{TX}} = 5{,}089$~TX positions and~$N_{\mathrm{RX}} = 21$ distributed~RX positions at~$z=2$\,m height, yielding~$N_{\mathrm{TX}} \times N_{\mathrm{RX}} = 106{,}869$ spectrum images in total.
We use an~80/20 random split~($4{,}071$ train~/~$1{,}018$ test~TX positions).

\begin{figure}[H]
\centering
\includegraphics[width=0.9\textwidth]{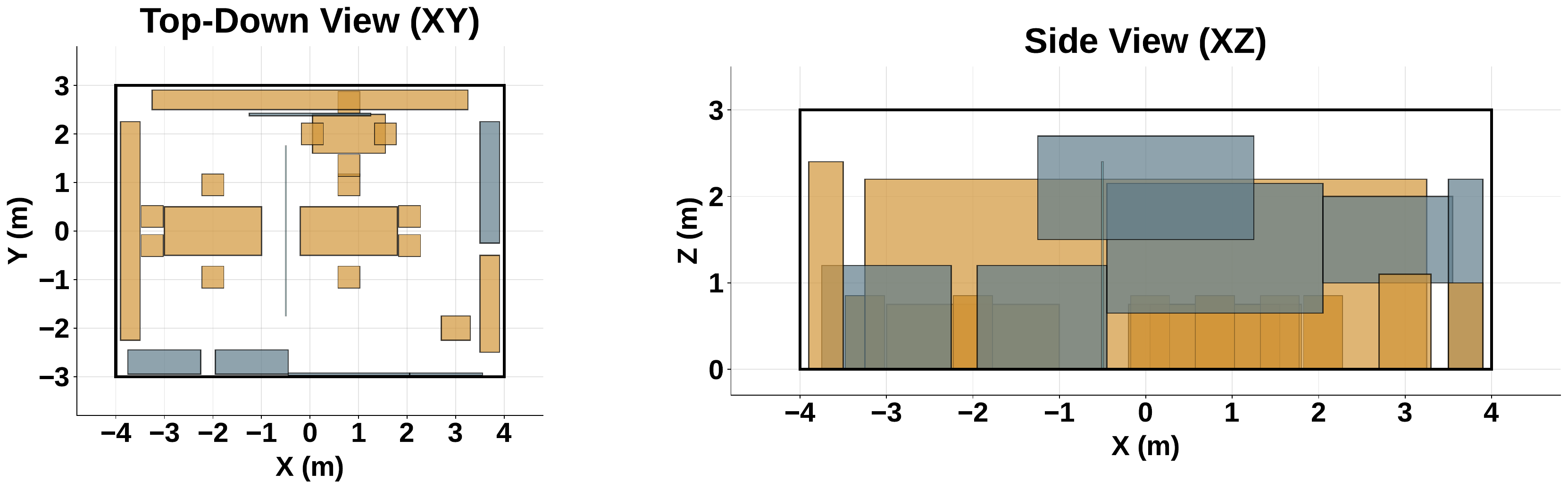}
\caption{Floor plan of the conference room~(8\,m $\times$ 6\,m $\times$ 3\,m) used for the~RFID spectrum dataset.}
\label{fig_layout}
\end{figure}

\textbf{Scene Construction.}
Each scene is modeled as a rectangular room with six surfaces~(floor, ceiling, four walls), built as~Mitsuba triangle meshes.
Furniture items including tables, chairs, bookshelves, metal cabinets, and glass partitions are represented as axis-aligned box meshes.
Each object is assigned an~ITU radio material~(\eg~concrete, plasterboard, wood, metal, glass) with calibrated relative permittivity, conductivity, and scattering coefficients following the~ITU-R~P.2040 standard.

\textbf{Channel Simulation.}
For each~TX--RX pair, we place a single isotropic~TX antenna and, at the~RX, an~$8\times8$ uniform planar array~(UPA) with half-wavelength spacing at~$f=2.4$\,GHz.
Sionna's~\texttt{PathSolver} performs ray tracing with up to~$5$ bounces, including line-of-sight, specular reflection, diffuse reflection, and refraction.
The solver returns~$\mathbf{H} \in \mathbb{C}^{N_{\mathrm{ant}} \times N_{\mathrm{sub}}}$ with~$N_{\mathrm{ant}}=64$ antennas and~$N_{\mathrm{sub}}=512$ subcarriers across~$100$\,MHz.

\textbf{Spectrum Computation.}
We compute the angular power spectrum via conventional beamforming~(CBF).
A~2D~Hanning window is applied across the array to suppress sidelobes.
Per-element phase and amplitude errors~($\pm 15^\circ$,~$\pm 1.5$\,dB) are introduced for realism.
The spatial covariance matrix is estimated as~$\mathbf{R}_{\mathrm{spec}} = \mathbf{H}\mathbf{H}^H / N_{\mathrm{sub}}$.
For each scan direction~$\left(\varphi, \theta\right)$ on a~$1^\circ$-resolution grid covering azimuth~$\left[0^\circ, 360^\circ\right)$ and elevation~$\left[0^\circ, 90^\circ\right)$, the beamforming power is:
\begin{equation}
P\left(\varphi, \theta\right) = \mathbf{a}^H\left(\varphi, \theta\right)\, \mathbf{R}_{\mathrm{spec}}\, \mathbf{a}\left(\varphi, \theta\right)
\end{equation}
where~$\mathbf{a}\left(\varphi, \theta\right) = \exp\!\left(j\, k\, \mathbf{p}_{\mathrm{ant}}^\top \hat{\mathbf{u}}\left(\varphi, \theta\right)\right) \in \mathbb{C}^{N_{\mathrm{ant}}}$ is the steering vector,~$k = 2\pi/\lambda$ is the wavenumber,~$\mathbf{p}_{\mathrm{ant}}$ contains the antenna element positions, and~$\hat{\mathbf{u}}$ is the unit direction vector.
The~3GPP~TR\,38.901 element radiation pattern is applied in the array's local coordinate frame.
The resulting spectrum is peak-normalized to~$0$\,dB and saved as a~$90 \times 360$ grayscale image.

\textbf{TX/RX Placement.}
We place~$N_{\mathrm{RX}} = 21$~RX positions at uniformly random~$(x, y)$ positions within the room with a minimum inter-RX spacing constraint.
TX positions are sampled randomly within the room volume and validated against furniture collisions via rejection sampling.

\subsection{WiFi CSI}
Prior work~\cite{zhao2023nerf2, wen2024wrfgs} uses the~Argos dataset~\cite{shepard2016understanding} for~CSI evaluation, where all antennas are co-located at a single position.
For multi-receiver evaluation, we use the~KU~Leuven~Ultra-Dense~Indoor~MaMIMO~CSI~Dataset~\NoHyper\cite{mamimo}\endNoHyper{}, collected using a~64-antenna massive~MIMO testbed operating at a center frequency of~$2.61$\,GHz with~$20$\,MHz bandwidth.
The original measurements comprise~$252{,}004$ channel snapshots over a~$1.25\,\text{m} \times 1.25\,\text{m}$ transmitter grid at~$5$\,mm spatial resolution, with~$100$ complex subcarriers per antenna.
In the distributed line-of-sight scenario, eight groups of eight antennas surround a~${\sim}5\,\text{m} \times 5\,\text{m}$ room in an octagonal layout at~$z=1$\,m.

\textbf{Preprocessing.}
We apply five preprocessing steps to the raw data.
First, phase calibration removes the random carrier phase offset per antenna, ensuring coherent combination across elements.
Second, the~$8$ antennas within each group are coherently averaged after calibration, yielding one effective~RX per group and~$N_{\mathrm{RX}}=8$ distributed~RXs in total.
Third, delay-domain denoising truncates the channel impulse response to the first~$10$ of~$100$ taps via~IFFT, truncation, and~FFT, suppressing noise-dominated late arrivals that carry no useful propagation information.
Fourth, the~$100$ subcarriers are averaged into~$N_{\mathrm{sc}}=26$ frequency bins to reduce redundancy while preserving the frequency-selective fading structure.
Fifth, nearby~TX positions are spatially averaged into~$N_{\mathrm{TX}}=6{,}000$ output positions to reduce wavelength-scale phase variation and produce a spatially consistent dataset suitable for scene-level learning.

\textbf{Final Dataset.}
The resulting dataset contains~$N_{\mathrm{TX}} = 6{,}000$~TX positions and~$N_{\mathrm{RX}} = 8$ distributed~RX positions, yielding~$N_{\mathrm{TX}} \times N_{\mathrm{RX}} = 48{,}000$ complex~CSI vectors of dimension~$52$~($N_{\mathrm{sc}}=26$ subcarriers $\times$~$2$ for real and imaginary), normalized by the global maximum magnitude.
We use an~80/20 random split~($4{,}800$ train~/~$1{,}200$ test~TX positions).

\section{Additional Experiments}
\label{app_additional_experiments}

This appendix presents additional experiments that complement the main text.
Appendix~\ref{app_per_receiver} breaks down Table~\ref{tab:overall} into per-receiver results.
Appendix~\ref{app_csi_vis} visualizes per-subcarrier~CSI traces.
Appendix~\ref{app_computation_cost} expands Table~\ref{tab:cost} into per-backbone training, inference, and storage breakdowns.
Appendix~\ref{app_spectrum_vis} shows spectrum visualizations on seen and unseen receivers.
Appendix~\ref{app_scaling_rx} analyzes how unseen-receiver accuracy scales with training receivers.
Appendix~\ref{app_ref_sensitivity} measures sensitivity to the~Stage~I reference receiver.
Appendix~\ref{app_network_planning} reports a network planning study.
Appendix~\ref*{app_receiver_gap} examines receiver-level accuracy gaps and their geometric correlates.
Appendix~\ref*{app_visibility} compares line-of-sight and non-line-of-sight pairs and assesses the occlusion probe.

\subsection{Per-Receiver Performance Breakdown}
\label{app_per_receiver}

Table~\ref{tab:overall} reports metrics averaged across all receivers, which can hide substantial per-receiver variation.
Per-receiver baselines fit each receiver in isolation, and aggregate metrics can obscure high-error receivers.
We therefore break the results down receiver-by-receiver on all three datasets, pairing each baseline with its receiver-conditioned counterpart that augments the same backbone with our conditioning module as a single shared model.

\begin{figure}[H]
\centering
\includegraphics[width=\textwidth]{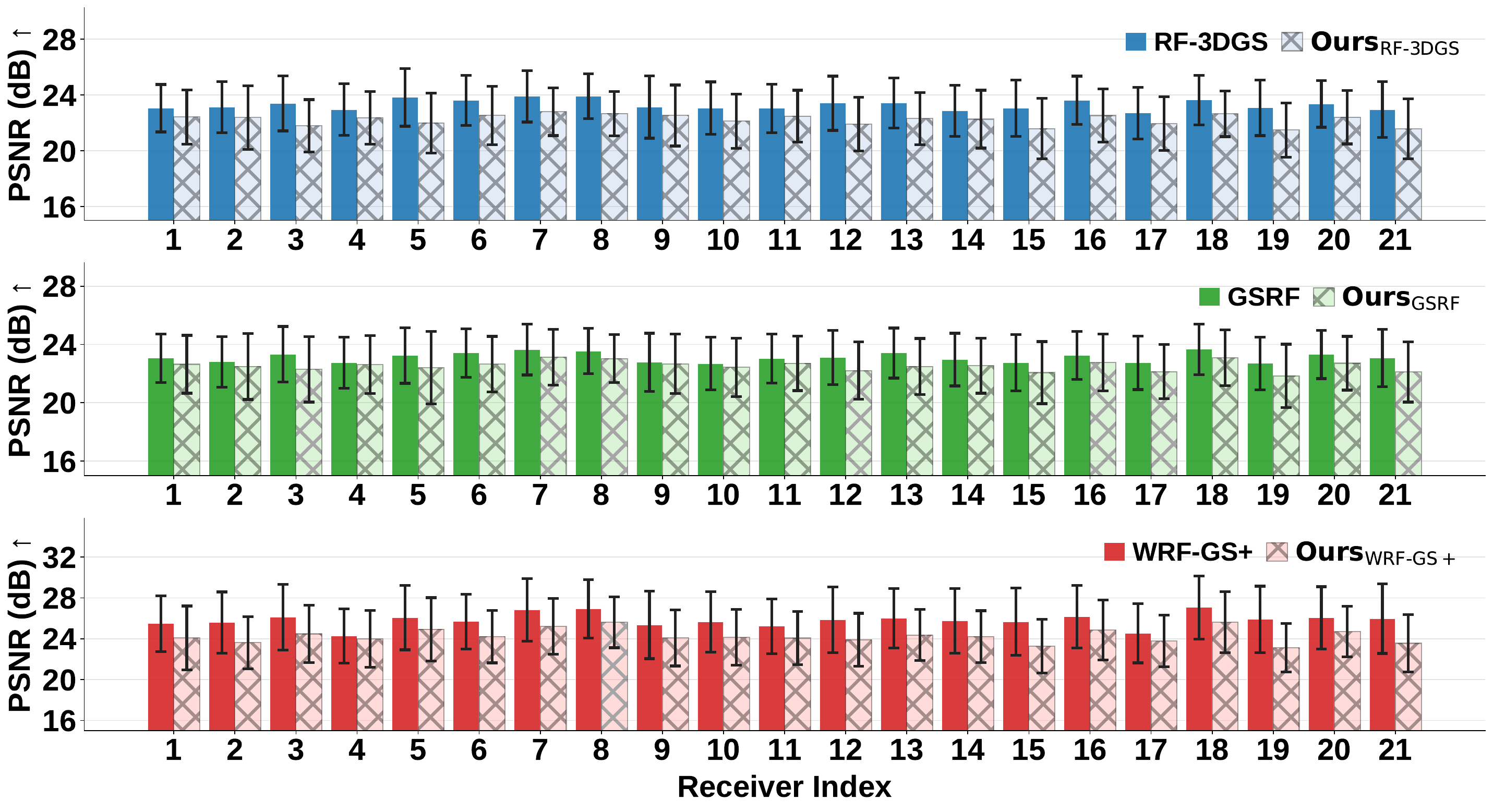}
\caption{Per-receiver~PSNR breakdown on~RFID spatial spectrum for each baseline and its receiver-conditioned counterpart~(higher is better).
One shared model covers all~$21$ receivers.
Table~\ref*{tab:overall} summarizes the corresponding average~PSNR.}
\label{fig:per_receiver_spectrum}
\end{figure}

\begin{figure}[H]
\centering
\includegraphics[width=\textwidth]{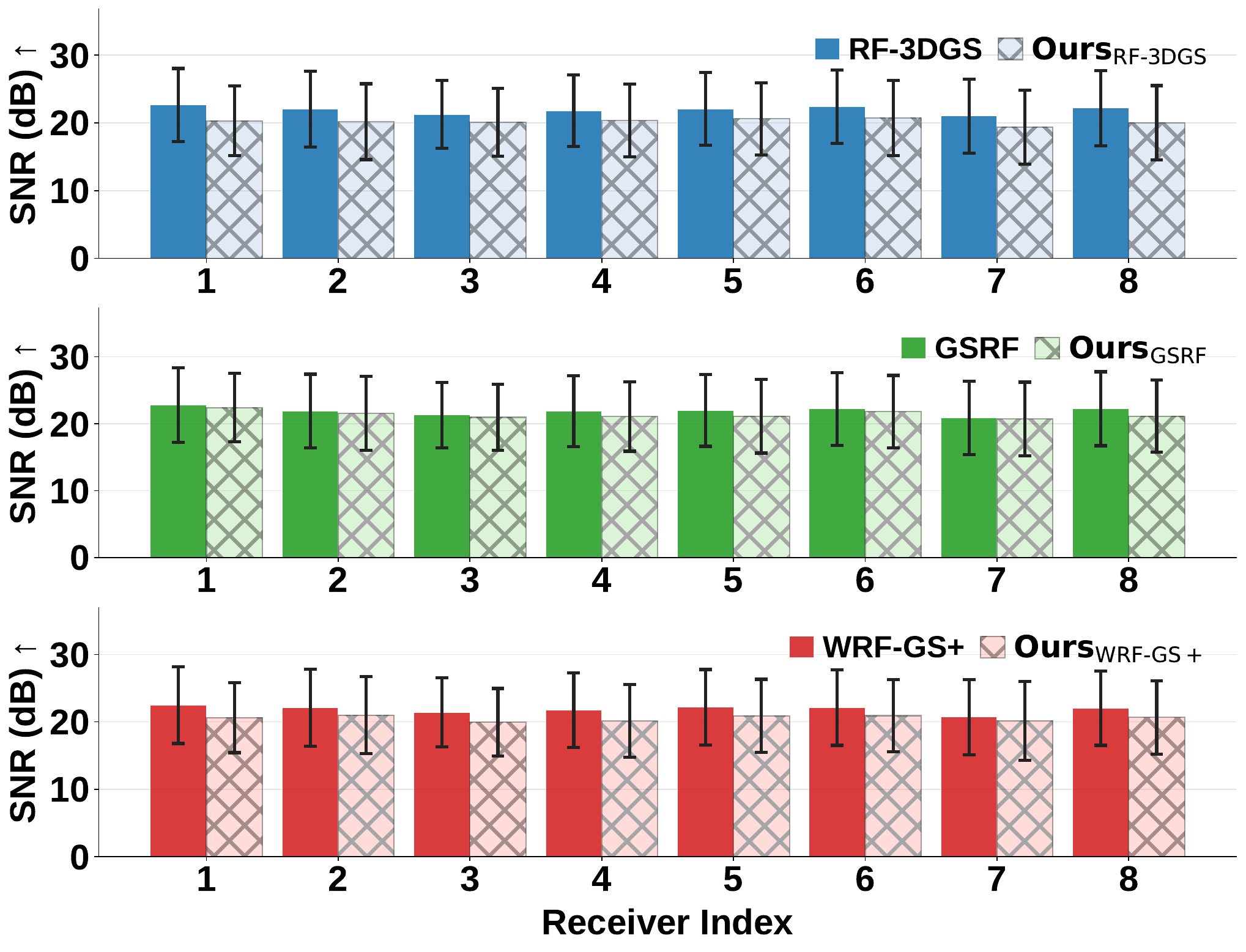}
\caption{Per-receiver~SNR breakdown on~WiFi~CSI for each baseline and its receiver-conditioned counterpart~(higher is better).
One shared model covers all~$8$ distributed receivers.
Table~\ref*{tab:overall} summarizes the corresponding average~SNR.}
\label{fig:per_receiver_csi}
\end{figure}

Figures~\ref{fig:per_receiver_ble}--\ref{fig:per_receiver_csi} compare all receivers across three baseline backbones,~RF-3DGS~\cite{zhang2024rf3dgs},~GSRF~\cite{yang2025gsrf}, and~WRF-GS+~\cite{wen2024wrfgs}, and their unified counterparts on~BLE~RSSI,~RFID spatial spectrum, and~WiFi~CSI respectively.
Four patterns emerge.

\textbf{Uniformity across receivers.}
For all three backbones and all three datasets, the receiver-conditioned variant yields a flatter metric profile across the receiver index.
This matches the behavior of a model that learns a transferable receiver-conditioning function rather than memorizing per-receiver mappings.
Weakly-covered receivers benefit from data collected at others.

\textbf{Outlier removal.}
On~BLE~RSSI, per-receiver~WRF-GS+ has two high-error receivers: receiver~$4$ at about~$28$\,dBm and receiver~$8$ at about~$19$\,dBm~MAE, both over~$5\!\times$ its median error.
The aggregate mean obscures their contribution.
Receiver-conditioned~WRF-GS+ removes both outliers.

\textbf{Spectrum and~CSI Accuracy.}
Table~\ref*{tab:overall} shows that the shared model generally tracks the per-receiver profile on~RFID spectrum and~WiFi~CSI.
Average~PSNR and~SNR are lower than those of specialized models by~$0.8$ to~$1.7$\,dB.

\textbf{Behavior Across Backbones.}
The conditioning design applies to all three tested~3DGS backbones, but its accuracy effect depends on the backbone and signal modality.
It improves the~BLE mean for~RF-3DGS and~GSRF, removes the isolated~WRF-GS+ failures while slightly raising its mean, and trades spectrum and~CSI accuracy for one shared model and lower deployment cost.

These observations corroborate the aggregated numbers in Table~\ref{tab:overall}.

\newpage
\subsection{Qualitative CSI Trace Visualization}
\label{app_csi_vis}

\noindent\textbf{Setting.}
We complement the aggregate~SNR numbers in the main text with per-subcarrier~CSI traces produced by our model on the~WiFi~CSI dataset.
We plot predicted and ground-truth~CSI vectors~$H_{\mathrm{Pred}},\, H_{\mathrm{GT}} \in \mathbb{C}^{26}$ at three randomly selected test transmitters.
The top row of Figure~\ref{fig:csi_traces} shows amplitude~$\left|H\left(f\right)\right|$ normalized per panel by~$\max\left(\left|H_{\mathrm{GT}}\right|,\, \left|H_{\mathrm{Pred}}\right|\right)$.
The bottom row shows phase wrapped to~$\left[0,\, 2\pi\right)$.

\textbf{Results.}
Across all three~TX positions, the predicted amplitude closely follows the ground-truth frequency response.
Each panel reproduces the multipath-induced amplitude peaks and nulls at the correct subcarrier indices.
The overall envelope is preserved with only minor shifts of~$\pm 1$ to~$2$ subcarriers near sharp nulls.
The phase traces track the ground-truth dispersion equally well.
The~$2\pi$ wraparound visible in~TX1 occurs at matching subcarrier indices for prediction and reference.
No panel shows a systematic bias or catastrophic drift.
This indicates that~\ourSystem captures both the magnitude and phase structure of the channel, not only amplitude statistics.

\begin{figure}[H]
\centering
\includegraphics[width=\textwidth]{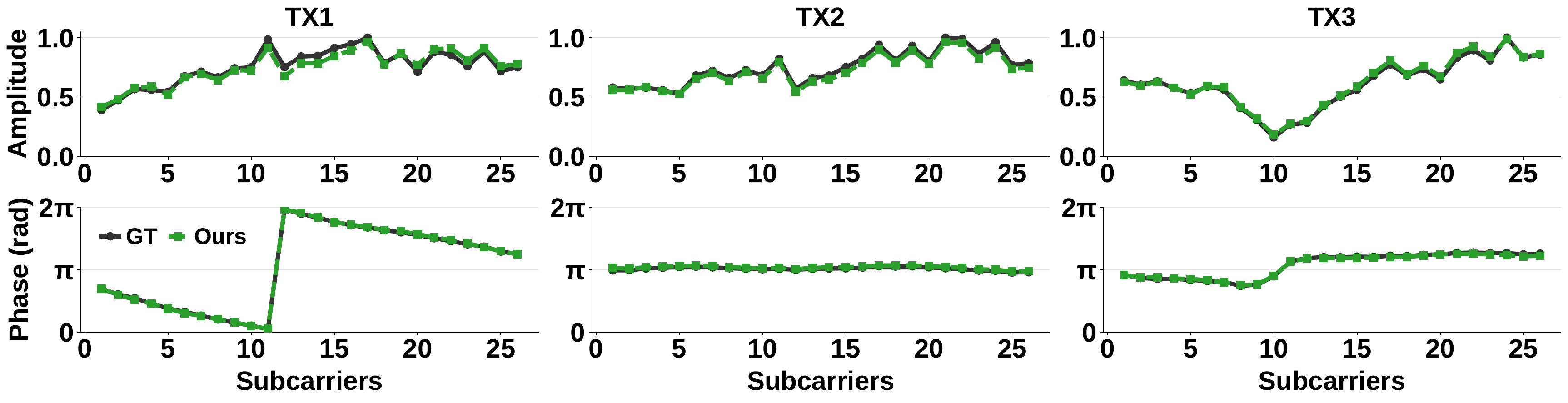}
\caption{Qualitative~CSI traces on the~WiFi~CSI dataset at three transmitter~(TX) positions.
Top row: amplitude~$\left|H\left(f\right)\right|$ normalized per panel.
Bottom row: phase wrapped to~$\left[0,\, 2\pi\right)$.
\ourSystem (green, dashed) tracks the ground truth (black, solid) in both amplitude and phase.}
\label{fig:csi_traces}
\end{figure}

\subsection{Computational Cost: Detailed Breakdown}
\label{app_computation_cost}

We expand on Table~\ref{tab:cost} and explain why the gains differ across backbones.

\noindent\textbf{Implicit versus explicit scene representations.}
Table~\ref*{tab:cost} shows different storage and compute profiles for implicit and explicit models.
NeRF$^2$ stores only~$2.6$\,MB but takes~$14.3$\,h to train and~$2436$\,ms to render~$N = 21$ spectra.
GSRF~$+$~\ourSystem stores~$26$\,MB yet trains in~$0.6$\,h and renders in~$31$\,ms.
The two paradigms swap their cost profiles:
\begin{itemize}[label=\textbullet, leftmargin=2em]
\item \textbf{Implicit methods (NeRF$^2$, GRaF)} encode the entire scene as a continuous function of position and direction using a small~MLP.
Each transmitter-receiver sample casts one ray per pixel and integrates by querying the~MLP at~$\mathcal{O}\!\left(64\text{--}256\right)$ depth points, totaling roughly~$2$\,M~MLP passes for a single~$90 \times 360$ spectrum, with the same number of backward passes during training.
Storage is small (only~MLP weights), but compute scales as~$\mathcal{O}\!\left(\text{rays} \cdot \text{depth\_samples} \cdot \text{MLP\_FLOPs}\right)$ and grows linearly with output resolution.
\item \textbf{Explicit methods (RF-3DGS, GSRF, WRF-GS+)} parameterize the scene as~$10$k to~$40$k Gaussian primitives, each carrying position, scale, rotation, opacity, and a directional coefficient vector with~$L = 16$ to~$100$ entries.
Storage scales as~$\mathcal{O}\!\left(K \cdot L\right)$ and dominates total bytes, but the~CUDA rasterizer projects all Gaussians once, sorts them per tile, and runs a single per-pixel alpha-blend pass with~$\mathcal{O}\!\left(\text{Gaussians overlapping the tile}\right)$ work, typically tens of~Gaussians per tile.
The pipeline is memory-bandwidth bound and parallelizes well on a~GPU.
For one transmitter, the~\ourSystem{} variants synthesize signals at all~$N=21$ receivers in~$31$ to~$43$\,ms.
\end{itemize}
Implicit methods use less storage but require more computation per sample, whereas explicit methods use more storage with lower per-sample computation.

\noindent\textbf{Training.}
\ourSystem totals~$90$k iterations per backbone:~$30$k for~Stage~I geometry pre-training and~$60$k for~Stage~II multi-receiver fitting.
The per-receiver baselines run~$N$ independent fits of~$T_{\text{base}}$ iterations each, with~$T_{\text{base}} = 200$k for~WRF-GS+ and~$30$k for~GSRF and~RF-3DGS.
The raw iteration ratio is~$N \cdot T_{\text{base}} / 90\text{k}$, with measured speedups also reflecting per-iteration cost differences.
\begin{itemize}[label=\textbullet, leftmargin=2em]
\item \textbf{RF-3DGS:} $2.8 \to 0.4$\,h, a~$7\times$ speedup that matches the raw iteration ratio of~$21 \cdot 30\text{k} / 90\text{k} = 7\times$ almost exactly, with similar per-iteration cost for~Stage~II and per-receiver training.
\item \textbf{GSRF:} $5.5 \to 0.6$\,h, a~$9\times$ speedup against the same~$7\times$ raw ratio.
The extra factor comes from~Stage~II running per-iteration slightly faster than per-receiver~GSRF, which we attribute to amortizing~FLE-basis evaluation across the batched receiver dimension.
\item \textbf{WRF-GS+:} $13.6 \to 0.3$\,h, a~$45\times$ speedup against a raw iteration ratio of~$21 \cdot 200\text{k} / 90\text{k} \approx 47\times$.
Multi-receiver overhead explains the roughly~$4\%$ difference.
The main gain comes from one~$90$k-iteration run instead of~$21$ separate~$200$k-iteration fits.
\end{itemize}
NeRF$^2$ and~GRaF are already shared models without~$N$-way training.
Their~$14$ to~$16$\,h costs reflect roughly~$600$k iterations of volumetric sampling and backpropagation, consistent with the implicit paradigm above.

\noindent\textbf{Inference.}
For the~GSRF and~RF-3DGS backbones, the multi-receiver rasterizer detailed in Appendix~\ref*{app_multirx_impl} renders all~$N$ receivers in one~CUDA call.
The per-method speedup over the per-receiver baseline depends on how much of the original cost was sharable versus per-pixel work:
\begin{itemize}[label=\textbullet, leftmargin=2em]
\item \textbf{GSRF:} $238 \to 31$\,ms, a~$7.6\times$ speedup that is the largest of the three because~FLE-degree~$9$ with~$L = 100$ components per Gaussian makes basis evaluation the dominant per-receiver cost.
Our batched~\texttt{eval\_fle} now performs this evaluation once per transmitter and reuses it across all~$N$ receivers, eliminating the most expensive sharable component.
\item \textbf{RF-3DGS:} $156 \to 35$\,ms, a~$4.4\times$ speedup.
At~SH degree~$4$ with~$L = 25$, basis evaluation is less costly, and most savings arise from shared preprocessing and tile sorting rather than the per-pixel kernel.
\item \textbf{WRF-GS+:} $52 \to 43$\,ms, only a~$1.2\times$ speedup.
In the measured implementation, the receiver-independent~DeformModel step is shared, whereas Gaussian rasterization runs separately for each receiver in~Python.
A native multi-receiver rasterizer offers a route to additional inference savings for this backbone.
\end{itemize}
The cost breakdown also explains two ordering effects.
First,~GSRF~$+$~\ourSystem ends up faster than~RF-3DGS~$+$~\ourSystem at~$31$ versus~$35$\,ms, despite per-receiver~GSRF being slower than per-receiver~RF-3DGS at~$11.3$ versus~$7.45$\,ms.
Once basis evaluation is amortized, the rasterizer dominates, and rasterizer cost scales with the number of~Gaussians, of which~RF-3DGS has more at~$37.7$k versus~$31.8$k.
Second,~NeRF$^2$ and~GRaF have no batchable structure across receivers.
Their per-transmitter cost remains~$N$ times the per-sample cost, yielding the~$60$ to~$1900\times$ differences reported in the table.

\noindent\textbf{Storage.}
Per-receiver baselines store~$N$ independent Gaussian-state checkpoints, whereas~\ourSystem stores one.
The deployable footprint decreases by roughly a factor of~$N$:
\begin{itemize}[label=\textbullet, leftmargin=2em]
\item \textbf{GSRF:} $548 \to 26$\,MB, a~$21\times$ reduction.
The absolute~\ourSystem size is the largest of the three because~GSRF carries~$K = 32$k Gaussians at~$L = 100$~FLE components, and storage scales as~$K \cdot L \cdot C \cdot 4$\,bytes plus roughly~$10K$ bytes for geometry attributes.
\item \textbf{RF-3DGS:} $182 \to 8.8$\,MB, a~$21\times$ reduction.
Although~RF-3DGS has more Gaussians than~GSRF~($K=38$k), its smaller~$L=25$ yields a smaller model.
\item \textbf{WRF-GS+:} $75.6 \to 3.2$\,MB, a~$24\times$ reduction.
The smallest absolute size comes from~$K = 14$k at~$L = 16$, plus a~$2.2$\,MB~DeformModel.
The slightly higher reduction ratio reflects the same~$N$ collapse at smaller absolute scale.
\end{itemize}
The~$3.2 \to 8.8 \to 26$\,MB ordering of the~\ourSystem rows therefore tracks~$K \cdot L$ rather than the per-receiver baseline ranking.
NeRF$^2$ and~GRaF are shared models that do not multiply with~$N$ but pay this saving back in inference compute.
Their sizes differ substantially~($2.6$\,MB for~NeRF$^2$ versus~$143.7$\,MB for~GRaF), because~GRaF stores training-receiver reference spectra.

\noindent\textbf{Receiver-count scaling.}
Per-receiver baselines require~$N$ separately trained and stored models.
The shared approach stores one model and uses a batched rasterization call for the~FLE and~SH backbones.
Stage~II training processes receiver-dependent samples, batched inference writes~$N$ outputs and performs some receiver-dependent work, and the~WRF-GS+ variant retains a render loop.
Model storage is shared, while training and inference costs vary with receiver count and backbone.
Table~\ref*{tab:cost} measures the resulting savings at~$N=21$.

\noindent\textbf{Receiver-Conditioned Scene~MLPs.}
NeRF$^2$~\cite{zhao2023nerf2} conditions a scene~MLP on the receiver, while~GRaF~\cite{yang2025generalizable} conditions on both endpoints.
This approach incurs per-ray volumetric sampling and~MLP evaluation at each query.
At the measured~$N=21$ setting, Table~\ref*{tab:cost} shows that~NeRF$^2$ and~GRaF take roughly~$60$ to~$1900\times$ longer per transmitter than the shared~3DGS variants.
\ourSystem retains the explicit~3DGS scene representation and applies receiver-conditioned modulation around frozen geometry.
The global and local branches require~$\mathcal{O}\!\left(L + K\right)$ network evaluations per receiver query, consistent with~\S\ref{sec_conditioning}.
The rasterizer in Appendix~\ref*{app_multirx_impl} shares transmitter-side projection and sorting across receivers.
The two-stage decomposition trains geometry once and restricts receiver dependence to radiance.

\subsection{Qualitative Spectrum Visualization}
\label{app_spectrum_vis}

Figure~\ref*{fig:spectrum_vis} visualizes spectrum predictions on one training receiver~(\textit{RX1}) and two receivers excluded from training~(\textit{RX2},~\textit{RX3}), each evaluated at three test transmitters~(\textit{TX1},~\textit{TX2},~\textit{TX3}).
For visual clarity we show one representative per design family:~\ourSystem{} for~Gaussian splatting with receiver-position conditioning,~{NeRF}$^2$ for continuous neural fields with receiver conditioning, and~WRF-GS+ for per-receiver~Gaussian splatting without conditioning.
GSRF and~RF-3DGS share the per-receiver structure of~WRF-GS+ and exhibit similar failures at unseen receivers.

\begin{figure}[H]
\centering
{\includegraphics[width=\textwidth]{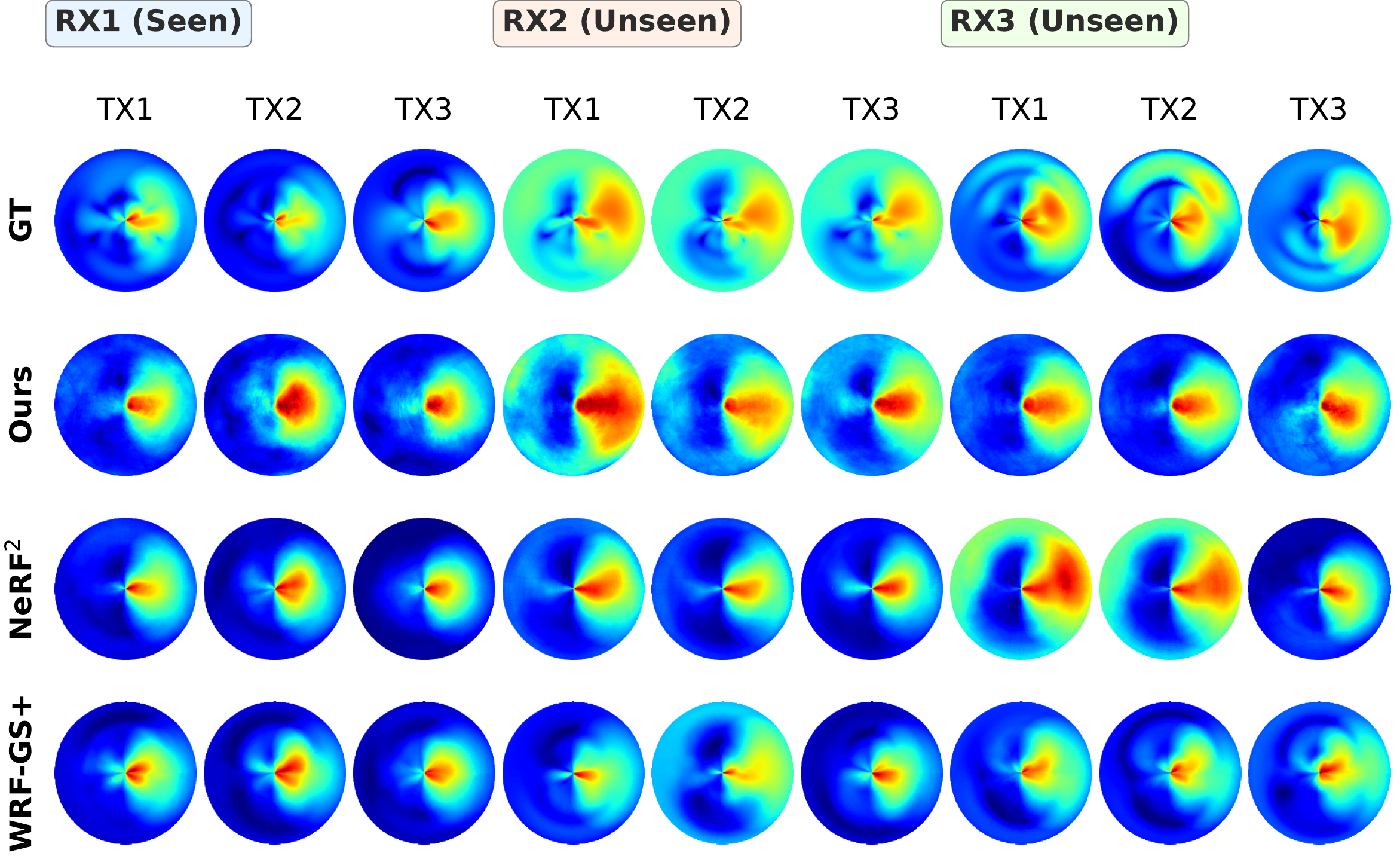}}
\caption{Qualitative comparison of~RFID spatial spectrum prediction.
\textit{RX1} is seen during training, while~\textit{RX2} and~\textit{RX3} are excluded from training.
\ourSystem{} produces the closest match to ground truth on unseen receivers.
{NeRF}$^2$ predictions are overly smooth, missing sharp directional peaks.
At receivers excluded from training,~WRF-GS+ uses the closest trained-receiver checkpoint, yielding degraded predictions.}
\label{fig:spectrum_vis}
\end{figure}

All methods recover the dominant structure on the seen receiver~\textit{RX1}.
Prediction differences are more pronounced at the unseen receivers~\textit{RX2} and~\textit{RX3}.

\textbf{\ourSystem{} Preserves Sharp Directional Peaks.}
On~\textit{RX2} and~\textit{RX3},~\ourSystem{} reproduces the narrow, bright lobes of the ground truth~(\eg~\textit{RX2--TX1},~\textit{RX2--TX2},~\textit{RX3--TX1}) at the correct azimuth and elevation.
Receiver conditioning transfers these peaks to unseen receivers.

\textbf{{NeRF}$^2$ Over-Smooths.}
{NeRF}$^2$ predictions on unseen receivers are plausible but visibly blurred.
Peaks smear into broader, lower-amplitude regions, and fine directional structure is lost.
The continuous~MLP interpolates between training receivers, while the high-frequency lobes that distinguish their views are attenuated in these examples.

\textbf{WRF-GS+ at Unseen Receivers.}
WRF-GS+ trains one model per receiver and lacks a query mechanism for new receiver positions.
At receiver positions excluded from training, it uses the nearest trained-receiver checkpoint.
The~\textit{RX2} and~\textit{RX3} predictions resemble a perturbed version of a different receiver's spectrum.
The coarse shape is sometimes close, but peak positions and amplitudes are systematically wrong.

\subsection{Scaling with Number of Receivers}
\label{app_scaling_rx}

\textbf{Setting.}
We measure how generalization to unseen receivers depends on receiver diversity seen during training.
For each~$N \in \{3, 6, 9, 12, 16, 21\}$, we train a single unified~\ourSystem model on~$N$ training receivers and evaluate on all~$21$ receivers, separating results into \emph{seen} and \emph{unseen} groups.
To control subset-selection variance, we draw~$7$ random training subsets at each~$N$ and report mean and standard deviation.
The~$N{=}21$ point serves as an upper bound where all receivers are seen.
Architecture, hyperparameters, and the transmitter train/test split remain fixed across runs.
Only the number of training receivers varies.

\textbf{Seen-Receiver Accuracy Is Flat in~$N$.}
The seen-receiver curves in Figure~\ref{fig_scaling_rx} stay flat as the training-receiver count grows.
On~BLE, seen~MAE varies only between~$3.55$ and~$3.80$\,dBm despite the training set growing~$7\!\times$.
On spectrum, seen~PSNR ranges from~$22.58$ to~$23.22$\,dB across the same receiver-count settings.
Across the tested subsets, shared geometry and receiver conditioning maintain similar seen-receiver accuracy as the number of jointly modeled receivers increases.

\begin{figure}[H]
\centering
\subfigure[BLE~RSSI]{
\label{fig_scaling_rx_a}
\includegraphics[width=.48\linewidth]{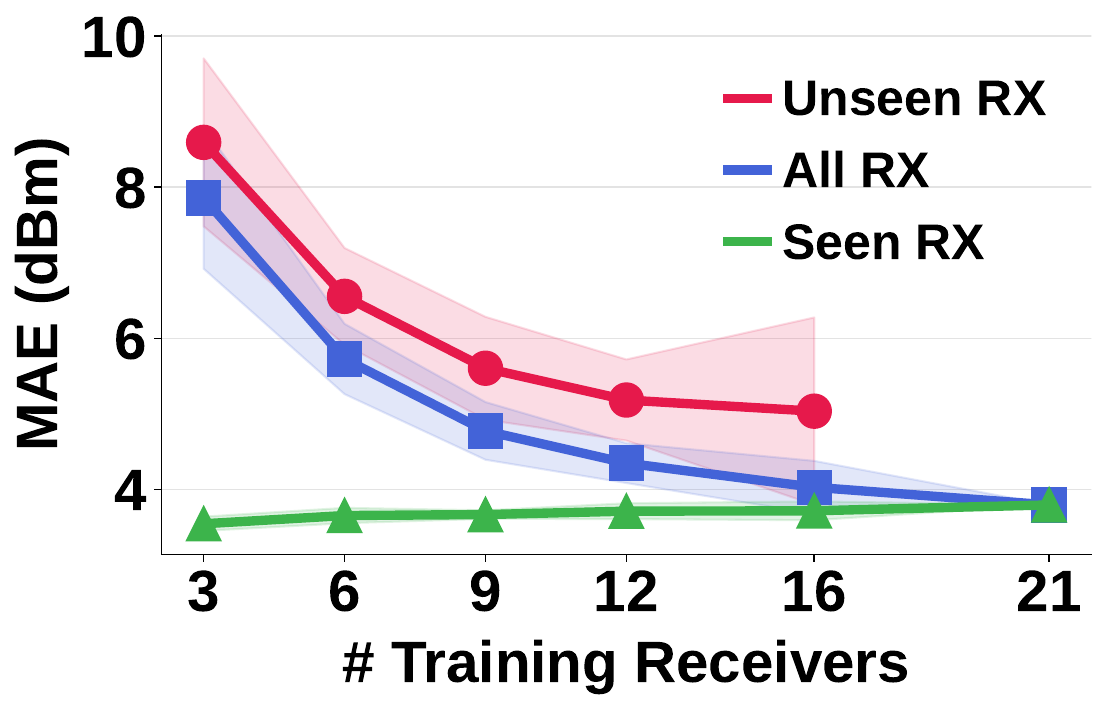}}
\subfigure[RFID~Spectrum]{
\label{fig_scaling_rx_b}
\includegraphics[width=.48\linewidth]{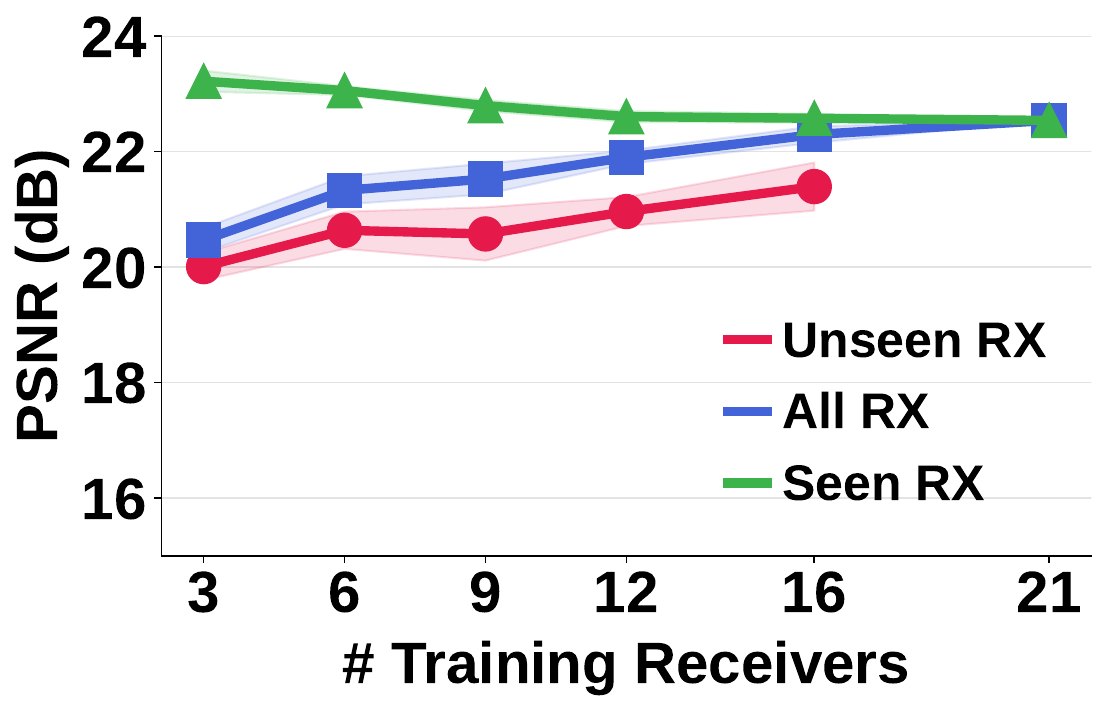}}
\caption{Scaling with the number of training receivers.
Seen accuracy stays flat while unseen accuracy improves with~$N$.}
\label{fig_scaling_rx}
\end{figure}

\textbf{Unseen-Receiver Accuracy Improves with~$N$.}
The unseen-receiver curves in Figure~\ref{fig_scaling_rx} drop substantially as~$N$ increases.
On~BLE, unseen~MAE drops from~$8.59$\,dBm at~$N{=}3$ to~$5.04$\,dBm at~$N{=}16$, a~$41\%$ relative reduction.
The seen-to-unseen gap shrinks from~$5.04$ to~$1.32$\,dBm over the same range.
On spectrum, unseen~PSNR climbs from~$20.01$\,dB at~$N{=}3$ to~$21.39$\,dB at~$N{=}16$, narrowing the gap from~$3.21$ to~$1.19$\,dB.
Both curves converge to the seen-receiver value as~$N\!\to\!21$.

\textbf{Implication.}
The combined pattern of flat seen accuracy and improving unseen accuracy indicates that the conditioning module learns a transferable, position-aware function of the receiver rather than memorizing per-receiver corrections.
A small number of training receivers leaves the conditioning manifold too sparse to disambiguate receiver-driven variation from scene-driven structure, producing large extrapolation errors.
Adding receivers densifies this manifold, enabling accurate interpolation to unseen positions.
Diminishing returns beyond~$N{=}12$ indicate that around a dozen training receivers cover much of the measured variation in the evaluated scenes.
For practical deployment, this means one can train on a modest number of representative anchor receivers rather than collecting calibration data at every possible position.

\newpage
\subsection{Sensitivity to Reference Receiver Choice}
\label{app_ref_sensitivity}

\textbf{Setting.}
Stage~I requires a reference receiver to learn the scene geometry.
Deployment of~\ourSystem requires selecting a reference receiver for~Stage~I.
We therefore evaluate the sensitivity of final performance to this selection.
We compare~$21$ single-receiver reference choices against an averaged-signal initialization that pools measurements across all receivers, training the full pipeline on the~BLE dataset under each choice.
Stage~II is trained across all~$21$ receivers with the same hyperparameters in every run, and we report per-receiver~MAE averaged across all~$21$ test receivers.

\textbf{Results.}
Figure~\ref{fig:ref_sensitivity} reports the resulting per-receiver~MAE.
The single-reference variants yield~$3.94 \pm 0.13$\,dBm across a range of~$\left[3.67, 4.22\right]$\,dBm, a~$0.55$\,dBm spread or about~$15\%$ of the best single-reference~MAE.
The narrow spread across the~$21$ tested choices is consistent with the receiver-independent geometry representation derived in Appendix~\ref*{app_decomposition}.
The averaged-signal initialization performs slightly better at~$3.57$\,dBm, since pooling across receivers averages out per-receiver measurement noise and yields a cleaner geometry-learning signal.
We therefore adopt the averaged-signal initialization as the default for all main-paper experiments.
Each single-reference run learns~Stage~I geometry using only one receiver without signal averaging, and its~MAE is~$0.10$ to~$0.65$\,dBm above the~$3.57$\,dBm averaged default.
Averaging is a refinement rather than a requirement and uses only measurements already available to~Stage~II.
In practice, however, either choice would yield comparable performance, and a practitioner can use whichever is most convenient for their deployment.

\begin{figure}[H]
\centering
{\includegraphics[width=\textwidth]{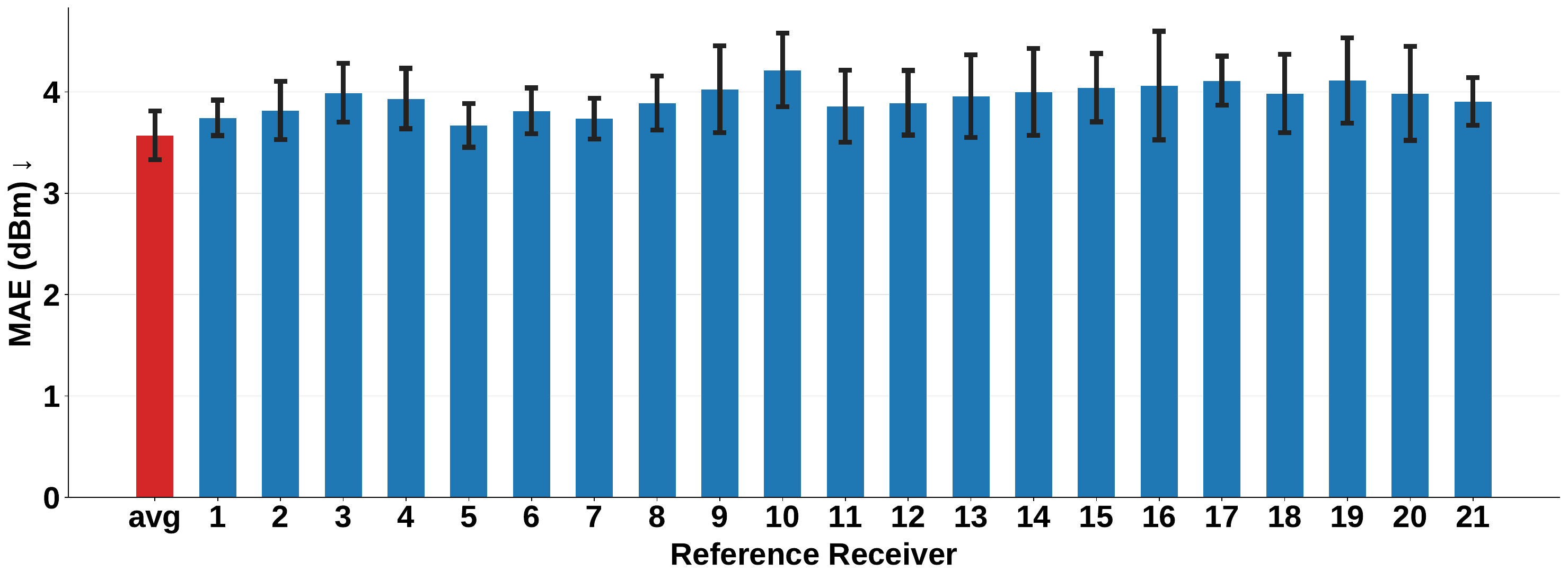}}
\caption{Sensitivity of~\ourSystem to the~Stage~I reference receiver choice on~BLE.
Each of the first~$21$ bars corresponds to using one receiver as the reference, and the~\texttt{avg} bar uses the averaged signal across all receivers.
The single-reference variants span~$3.67$ to~$4.22$\,dBm, a~$0.55$\,dBm spread, while the averaged-signal initialization achieves~$3.57$\,dBm.}
\label{fig:ref_sensitivity}
\end{figure}

\subsection{Application Study: Network Planning}
\label{app_network_planning}

\noindent\textbf{Setting.}
A common task for a building owner is to choose where to install a small set of access points so that wireless coverage is maximized across all client positions.
Without an accurate channel model, the planner is restricted to deploying access points at positions already surveyed.
With~\ourSystem, the planner can score unmeasured candidates and select higher-coverage sites.

As in~\S\ref*{sec_generalization}, we use a~$3$-fold split of the~$21$-gateway~BLE~RSSI dataset into~$14$ trained gateways and~$7$ gateways excluded from training.
For each fold we randomly sample~$K = 5$ surveyed gateways from the~$14$ trained ones and compare three deployment strategies.
\textsc{Survey-Only} deploys at the~$5$ surveyed positions with no model and no choice.
\textsc{Ours} runs greedy maximum-coverage selection over all~$21$ candidate positions, using real~RSSI for the~$5$ surveyed gateways and~\ourSystem predictions for the other~$16$.
\textsc{Upper Bound} runs the same greedy selection but uses real~RSSI at all~$21$ gateways, giving a full-survey oracle.
Once a deployment is chosen, we evaluate it on the test transmitters by computing each transmitter's strongest real~RSSI among the~$5$ deployed access points, which is the standard max-power association in wireless networks.
Figure~\ref{fig:network_planning} renders one representative fold as a~$24 \times 32$ grid heatmap with cell size of about~$1.6 \times 1.6$\,m.
Each cell is colored by the fraction of transmitters whose strongest-AP~RSSI exceeds~$-80$\,dBm, and cells with no test transmitters are shown in gray.

\begin{figure}[H]
\centering
\includegraphics[width=\textwidth]{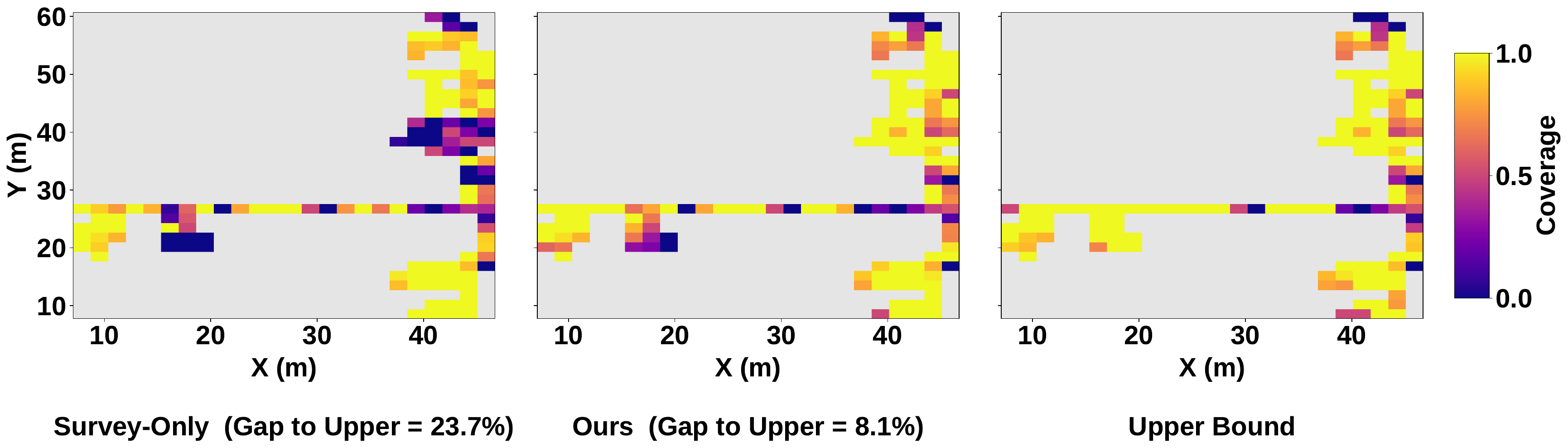}
\caption{Coverage heatmaps for~$K = 5$ access-point deployment on the~BLE dataset.
Each cell shows the fraction of test transmitters in a~$\sim 1.6 \times 1.6$\,m region whose strongest-AP~RSSI exceeds~$-80$\,dBm.
Gray cells contain no measurements.
\textsc{Survey-Only} deploys at~$5$ surveyed gateways, \textsc{Ours} uses~\ourSystem predictions to select~$5$ from all~$21$ candidates, and~\textsc{Upper Bound} greedily selects~$5$ from full-survey~RSSI.
For the displayed fold,~\textsc{Ours} reduces the gap to~\textsc{Upper Bound} from~$23.7\%$ to~$8.1\%$.}
\label{fig:network_planning}
\end{figure}

\noindent\textbf{Results.}
The three panels share the same building footprint and the same evaluation transmitters, yet the resulting coverage patterns differ markedly.
\textsc{Survey-Only} leaves a cluster of low-coverage~(dark purple) cells in the upper corridor, where the five randomly surveyed access points provide insufficient coverage.
These coverage gaps reflect the restriction to surveyed access-point positions.
\textsc{Ours} reduces the low-coverage region by selecting higher-coverage sites among the~$16$ unmeasured candidates in place of lower-coverage surveyed sites.
The remaining low-coverage cells resemble the full-survey reference.

We quantify similarity to the oracle by the average per-cell absolute coverage difference, denoted gap-to-Upper.
For the representative fold in Figure~\ref*{fig:network_planning}, the gap falls from~$23.7\%$ under~\textsc{Survey-Only} to~$8.1\%$ under~\textsc{Ours}, a roughly~$3\times$ reduction.
Aggregating over the~$3$ folds, \textsc{Ours} raises the share of transmitters covered above~$-80$\,dBm from~$87.9\%$ under~\textsc{Survey-Only} to~$92.6\%$, recovering about~$60\%$ of the gap to the full-survey upper bound of~$95.8\%$.
In this study, predicted~RSSI expands the candidate set while retaining the same five-gateway survey budget.

\begingroup

\noindent\textbf{Practical Survey Guidance.}
The measured-to-virtual balance depends on anchor diversity, synthesis fidelity, learnability of receiver effects, and how well the model represents the scene.
A practical survey begins with spatially distributed gateways spanning distinct propagation conditions, then uses virtual gateways to densify candidates within the surveyed region.
Withholding some measured gateways during model selection can reveal weakly supported areas and guide additional measurements until unseen-gateway accuracy begins to saturate.
Figure~\ref*{fig_scaling_rx} shows diminishing returns beyond about~$12$ gateways in this scene.
Jointly choosing real and virtual sites remains open.
\endgroup

\begingroup
\subsection{Receiver-Level Accuracy Gaps}
\label{app_receiver_gap}
\endgroup

\begingroup

The~RFID spectrum simulation provides visibility labels for examining how the accuracy gap varies with receiver geometry.
For each of the~$21$ receivers, we compare~\withours{GSRF} with the corresponding per-receiver~GSRF baseline on the same test samples.
We relate the receiver-level~PSNR gap to the fraction of transmitters with non-line-of-sight~(NLoS) paths, mean grid-probed occlusion, and distance to the nearest training receiver under the protocol of~\S\ref*{sec_generalization}.
For seen receivers,~NLoS fraction has the largest observed association with the gap (Pearson~$r=0.43$,~$p=0.054$).
Two corner receivers behind furniture,~\textit{rx15} and~\textit{rx19}, have approximately~$99\%$~NLoS pairs and the largest gaps.
The correlation is~$r=0.27$ after excluding those two receivers.
At unseen receivers, the correlations are~$r=0.23$ for~NLoS fraction,~$r=0.12$ for mean probed occlusion, and~$r=0.04$ (~$p=0.88$) for distance to the nearest training receiver.
Nearest-anchor distance shows little association with the gap in this split.
\endgroup

\begingroup
\subsection{Visibility, Bounce Depth, and Occlusion Features}
\label{app_visibility}
\endgroup

\begingroup

We use~Sionna path geometry to split the same~RFID test pairs into line-of-sight~(LoS) and non-line-of-sight~(NLoS) groups according to last-segment visibility.
Table~\ref*{tab:visibility_split} reports~PSNR for each group.
For the unseen-receiver comparison, the per-receiver~GSRF reference is trained separately with measurements from each evaluated receiver.
\endgroup

\begin{table}[H]
\centering

\makeatletter\def\fnum@table{Table~\thetable}\makeatother
\caption{RFID spectrum~PSNR~(dB) by last-segment visibility.
At unseen receivers, the per-receiver baseline is trained on measurements from each evaluated receiver.}
\label{tab:visibility_split}
\begin{tabular}{lcc}
\toprule
Method & LoS Pairs & NLoS Pairs \\
\midrule
Per-receiver~GSRF & $23.35$ & $23.20$ \\
\withours{GSRF}, seen receivers & $22.35$ & $22.59$ \\
\withours{GSRF}, unseen receivers & $21.07$ & $21.40$ \\
\bottomrule
\end{tabular}
\end{table}

\begingroup

Relative to the per-receiver model on the same pairs, the shared model trails by~$1.00$\,dB on line-of-sight~(LoS) and~$0.61$\,dB on~NLoS pairs at seen receivers.
The corresponding gaps at unseen receivers are~$2.28$ and~$1.80$\,dB.
Among~NLoS pairs, even the heaviest blockage-severity tercile has gaps of~$0.64$\,dB seen and~$1.50$\,dB unseen, both below the respective~LoS gaps.
Grouping pairs instead by dominant-path bounce count gives the largest unseen gap for direct paths at~$2.2$\,dB.
The one-, two-, and three-or-more-bounce groups have unseen gaps between~$1.6$ and~$2.0$\,dB, and the seen groups show no increasing trend with bounce count.
In this test set, greater blockage and bounce depth have gaps comparable to or smaller than those of~LoS and direct paths, respectively.
The~PSNR gaps are relative to each group's signal dynamic range.

We also compare the grid-probed last-segment transmittance~$T_k$ with~Sionna ray-traced visibility.
At unseen receivers the probe gives Pearson~$r=0.40$ and an area under the receiver operating characteristic curve~(AUC) of~$0.70$.
At seen receivers it gives~$r=0.40$ and~AUC~$0.73$.
The probe provides similar, moderate visibility discrimination at seen and unseen receivers.
\endgroup

\begingroup
\subsection{Occupancy-Grid Resolution}
\label{app_grid_resolution}
\endgroup

\begingroup

We retrain~Stage~II on~RFID spectrum at each occupancy-grid resolution, keeping the frozen~Stage~I geometry fixed.
Table~\ref*{tab:grid_resolution} shows that all three settings remain within~$0.1$\,dB of the default~$128^3$ grid.
Across the evaluated~$64^3$ to~$256^3$ range,~PSNR varies by only~$0.04$\,dB, indicating limited sensitivity to grid size for this room.
\endgroup

\begin{table}[H]
\centering

\makeatletter\def\fnum@table{Table~\thetable}\makeatother
\caption{Occupancy-grid resolution on~RFID spectrum.
Stage~II is retrained for each setting.}
\label{tab:grid_resolution}
\begin{tabular}{lccc}
\toprule
Grid & $64^3$ & $128^3$ (default) & $256^3$ \\
\midrule
PSNR (dB) & $22.43$ & $22.47$ & $22.43$ \\
\bottomrule
\end{tabular}
\end{table}

\begingroup
\subsection{Conditioning Capacity}
\label{app_capacity}
\endgroup

\begingroup

We vary the widths of both conditioning~MLPs, including the local branch, while keeping their depth at three layers.
We retrain~Stage~II for each~RFID setting using the same frozen~Gaussian geometry.
Table~\ref*{tab:conditioning_width} reports nearly constant~PSNR across a roughly~$25\times$ increase in conditioning parameters.
The total~PSNR spread is~$0.04$\,dB, and every width remains within~$0.1$\,dB of the default.
Stage~II takes about~$21$ minutes per setting, while inference for~$N=21$ receivers remains about~$31$\,ms because the same Gaussian set is rasterized.
Widening the three-layer conditioner has little effect on the approximately~$0.8$\,dB seen-receiver spectrum gap to per-receiver~GSRF.
The~CSI configuration already uses width~$256$, four times the spectrum default, and has a comparable gap.
Conditioner depth and alternative designs remain directions for improving this accuracy--efficiency tradeoff.
\endgroup

\begin{table}[H]
\centering

\makeatletter\def\fnum@table{Table~\thetable}\makeatother
\caption{RFID spectrum performance across conditioning widths with three-layer branches retrained in~Stage~II.}
\label{tab:conditioning_width}
\begin{tabular}{lcccc}
\toprule
Width~$d$ & $64$ (default) & $128$ & $256$ & $512$ \\
\midrule
PSNR (dB) & $22.47$ & $22.48$ & $22.44$ & $22.46$ \\
Conditioning parameters & $11.7$k & $28.1$k & $85.6$k & $298.9$k \\
\bottomrule
\end{tabular}
\end{table}
\section{Four Evaluation Configurations}\label{app_generalization_configs}

RF data synthesis is queried along two axes, transmitter and receiver.
Evaluation configurations are distinguished by which axis contains test positions excluded from training.
Table~\ref{tab:gen_configs} enumerates the four possibilities.
The labels in this table follow our experimental setup: transmitter positions are densely sampled, while receivers occupy a smaller set of fixed positions.

\begin{table}[h]
\centering
\caption{Four evaluation configurations.}
\label{tab:gen_configs}
\begin{tabular}{l|l|l}
\toprule
& \textbf{Seen Transmitter} & \textbf{Unseen Transmitter} \\
\midrule
\textbf{Seen Receiver} & (1) training fit & (2) novel transmitter \\
\midrule
\textbf{Unseen Receiver} & (3) novel receiver & (4) novel transmitter~$\&$ receiver \\
\bottomrule
\end{tabular}
\end{table}

Prior per-receiver~3DGS methods principally evaluate cell~(2), predicting at an unseen position on the moving, densely surveyed side while retaining a checkpoint trained for the fixed receiver side~\NoHyper\cite{zhang2024rf3dgs, yang2025gsrf, wen2024wrfgs}\endNoHyper{}.
The endpoint names differ across papers:~WRF-GS+ fixes the physical receiver, whereas~RF-3DGS and~GSRF fix the physical transmitter, which maps to the fixed receiver role in our formulation after applying channel reciprocity.
Here, reciprocity aligns endpoint roles for comparing model configurations, while direction-resolved measurements such as angular spectra retain their physical endpoint orientation.
In all three designs, the moving side enters rendering, while the fixed side is absorbed into model parameters or training data.
NeRF$^2$ instead queries one shared field from the receiver's ray origin and can support unseen receivers at greater computational cost~\NoHyper\cite{zhao2023nerf2}\endNoHyper{}.
GRaF's query uses reference spectra measured at each receiver~\NoHyper\cite{yang2025generalizable}\endNoHyper{}.
We follow this convention in~\S\ref*{sec_overall} for direct comparison, while using one unified model for all~$N$ receivers.
Table~\ref*{tab:cost} reports the resulting training, inference, and storage savings.

Cell~(4) motivates~\ourSystem: it predicts at an unseen transmitter and a fixed receiver excluded from training but located within the surveyed scene and deployment region.
This models the addition of a receiver within an existing surveyed environment without collecting its measurements or fitting another receiver-specific checkpoint.
We report cell~(4) in~\S\ref{sec_generalization}.
Cell~(4) subsumes cell~(3), since reusing training transmitters at an unseen receiver is a sub-case of generalizing along both axes.
Cell~(1) reports training-time fit and is bounded by cell~(2).

\section{Broader Impacts}
\label{app_broader_impacts}

\ourSystem reduces the cost of collecting~RF data for wireless networking and sensing.
Site surveys at thousands of transmitter--receiver positions are labor-intensive and time-consuming.
\mbox{Synthesizing}~RF data at unmeasured fixed receiver positions inside a static, surveyed environment could reduce the measurement burden for wireless planning and indoor localization, as illustrated in~\S\ref*{sec_application} and Appendix~\ref*{app_network_planning}.
The coverage estimates may also guide infrastructure placement.

A future application is connectivity-aware planning for mobile agents within a static, previously surveyed environment.
A shared model could score candidate receiver positions inside the surveyed region without fitting a new checkpoint at each location.
Future work will evaluate continuous receiver motion, changing environments, and robotic deployments.

The same capability could be misused for passive localization or tracking without consent if an attacker combines synthesized fingerprints with captured signals.
This risk is not unique to~\ourSystem and applies to any~RF fingerprinting method, but practitioners deploying fingerprint databases should follow standard privacy safeguards such as on-device inference, aggregation, and informed consent.

\clearpage
\section*{NeurIPS Paper Checklist}

The checklist is designed to encourage best practices for responsible machine learning research, addressing issues of reproducibility, transparency, research ethics, and societal impact. Do not remove the checklist: {\bf The papers not including the checklist will be desk rejected.} The checklist should follow the references and follow the (optional) supplemental material.  The checklist does NOT count towards the page
limit. 

Please read the checklist guidelines carefully for information on how to answer these questions. For each question in the checklist:
\begin{itemize}
    \item You should answer \answerYes{}, \answerNo{}, or \answerNA{}.
    \item \answerNA{} means either that the question is Not Applicable for that particular paper or the relevant information is Not Available.
    \item Please provide a short (1--2 sentence) justification right after your answer (even for \answerNA). 
\end{itemize}

{\bf The checklist answers are an integral part of your paper submission.} They are visible to the reviewers, area chairs, senior area chairs, and ethics reviewers. You will also be asked to include it (after eventual revisions) with the final version of your paper, and its final version will be published with the paper.

The reviewers of your paper will be asked to use the checklist as one of the factors in their evaluation. While \answerYes{} is generally preferable to \answerNo{}, it is perfectly acceptable to answer \answerNo{} provided a proper justification is given (e.g., error bars are not reported because it would be too computationally expensive'' or ``we were unable to find the license for the dataset we used''). In general, answering \answerNo{} or \answerNA{} is not grounds for rejection. While the questions are phrased in a binary way, we acknowledge that the true answer is often more nuanced, so please just use your best judgment and write a justification to elaborate. All supporting evidence can appear either in the main paper or the supplemental material, provided in appendix. If you answer \answerYes{} to a question, in the justification please point to the section(s) where related material for the question can be found.

IMPORTANT, please:
\begin{itemize}
    \item {\bf Delete this instruction block, but keep the section heading ``NeurIPS Paper Checklist"},
    \item  {\bf Keep the checklist subsection headings, questions/answers and guidelines below.}
    \item {\bf Do not modify the questions and only use the provided macros for your answers}.
\end{itemize}


\begin{enumerate}

\item {\bf Claims}
    \item[] Question: Do the main claims made in the abstract and introduction accurately reflect the paper's contributions and scope?
    \item[] Answer: \answerYes{}
    \item[] Justification: The abstract and introduction reflect the main contributions of~\ourSystem, including receiver-generalizable RF data synthesis within a single model and a radiance-representation-agnostic conditioning design.

\item {\bf Limitations}
    \item[] Question: Does the paper discuss the limitations of the work performed by the authors?
    \item[] Answer: \answerYes{}
    \item[] Justification: Section~\ref{sec_conclusion} discusses the limitations of~\ourSystem, including the static-scene assumption and the lack of cross-scene generalization.

\item {\bf Theory assumptions and proofs}
    \item[] Question: For each theoretical result, does the paper provide the full set of assumptions and a complete (and correct) proof?
    \item[] Answer: \answerYes{}
    \item[] Justification: The geometry-radiance decomposition and the last-segment factorization are stated with full assumptions in Section~\ref{sec_method} and proved in Appendices~\ref{app_decomposition},~\ref{app_factorization}, and~\ref{app_corollary_proofs}.

    \item {\bf Experimental result reproducibility}
    \item[] Question: Does the paper fully disclose all the information needed to reproduce the main experimental results of the paper to the extent that it affects the main claims and/or conclusions of the paper (regardless of whether the code and data are provided or not)?
    \item[] Answer: \answerYes{}
    \item[] Justification: Appendix~\ref{app_implementation} provides full implementation details and hyperparameters, and Appendix~\ref{app_dataset} describes the datasets and splits used for all experiments.

\item {\bf Open access to data and code}
    \item[] Question: Does the paper provide open access to the data and code, with sufficient instructions to faithfully reproduce the main experimental results, as described in supplemental material?
    \item[] Answer: \answerYes{}
    \item[] Justification: The BLE RSSI and WiFi CSI datasets are publicly available as cited, the RFID spectrum dataset generation is fully described in Appendix~\ref{app_dataset}, and the code will be released upon publication.

\item {\bf Experimental setting/details}
    \item[] Question: Does the paper specify all the training and test details (e.g., data splits, hyperparameters, how they were chosen, type of optimizer) necessary to understand the results?
    \item[] Answer: \answerYes{}
    \item[] Justification: Appendix~\ref{app_implementation} specifies the optimizer, learning rates, and all hyperparameters, and Appendix~E specifies the data splits for each dataset.

\item {\bf Experiment statistical significance}
    \item[] Question: Does the paper report error bars suitably and correctly defined or other appropriate information about the statistical significance of the experiments?
    \item[] Answer: \answerYes{}
    \item[] Justification: Table~\ref{tab:overall} and all per-receiver results report the mean and standard deviation across per-receiver means, as described in Section~\ref{sec_evaluation}.

\item {\bf Experiments compute resources}
    \item[] Question: For each experiment, does the paper provide sufficient information on the computer resources (type of compute workers, memory, time of execution) needed to reproduce the experiments?
    \item[] Answer: \answerYes{}
    \item[] Justification: Section~\ref{sec_evaluation} states that all experiments are conducted on a single NVIDIA H100 GPU, and Table~\ref{tab:cost} reports per-method training and inference times.
    
\item {\bf Code of ethics}
    \item[] Question: Does the research conducted in the paper conform, in every respect, with the NeurIPS Code of Ethics \url{https://neurips.cc/public/EthicsGuidelines}?
    \item[] Answer: \answerYes{}
    \item[] Justification: The research conforms with the NeurIPS Code of Ethics in every respect.

\item {\bf Broader impacts}
    \item[] Question: Does the paper discuss both potential positive societal impacts and negative societal impacts of the work performed?
    \item[] Answer: \answerYes{}
    \item[] Justification: Appendix~\ref{app_broader_impacts} discusses the broader impacts of~\ourSystem.
    
\item {\bf Safeguards}
    \item[] Question: Does the paper describe safeguards that have been put in place for responsible release of data or models that have a high risk for misuse (e.g., pre-trained language models, image generators, or scraped datasets)?
    \item[] Answer: \answerNA{}
    \item[] Justification: The paper poses no such risks.

\item {\bf Licenses for existing assets}
    \item[] Question: Are the creators or original owners of assets (e.g., code, data, models), used in the paper, properly credited and are the license and terms of use explicitly mentioned and properly respected?
    \item[] Answer: \answerYes{}
    \item[] Justification: All existing datasets and baselines used in the paper are properly cited in Section~\ref{sec_evaluation} and Appendix~\ref{app_dataset}.

\item {\bf New assets}
    \item[] Question: Are new assets introduced in the paper well documented and is the documentation provided alongside the assets?
    \item[] Answer: \answerYes{}
    \item[] Justification: The RFID spatial spectrum dataset is documented in Appendix~\ref{app_dataset}, and the implementation is documented in Appendix~\ref{app_implementation}.

\item {\bf Crowdsourcing and research with human subjects}
    \item[] Question: For crowdsourcing experiments and research with human subjects, does the paper include the full text of instructions given to participants and screenshots, if applicable, as well as details about compensation (if any)? 
    \item[] Answer: \answerNA{}
    \item[] Justification: The paper does not involve crowdsourcing or research with human subjects.

\item {\bf Institutional review board (IRB) approvals or equivalent for research with human subjects}
    \item[] Question: Does the paper describe potential risks incurred by study participants, whether such risks were disclosed to the subjects, and whether Institutional Review Board (IRB) approvals (or an equivalent approval/review based on the requirements of your country or institution) were obtained?
    \item[] Answer: \answerNA{}
    \item[] Justification: The paper does not involve research with human subjects.

\item {\bf Declaration of LLM usage}
    \item[] Question: Does the paper describe the usage of LLMs if it is an important, original, or non-standard component of the core methods in this research? Note that if the LLM is used only for writing, editing, or formatting purposes and does \emph{not} impact the core methodology, scientific rigor, or originality of the research, declaration is not required.
    \item[] Answer: \answerNA{}
    \item[] Justification: The core method development does not involve LLMs.

\end{enumerate}

\end{document}